\documentclass[conference]{IEEEtran-modified}
\usepackage[english]{babel}
\usepackage[latin1]{inputenc}
\usepackage{graphicx}
\usepackage{subfigure}
\usepackage{amssymb}
\usepackage{amsmath}
\usepackage{verbatim}
\usepackage{cite}
\usepackage{newalg}
\usepackage{url}

\begin{document}

\newtheorem{thm}{Theorem}
\newtheorem{cor}{Corollary}
\newtheorem{lem}{Lemma}
\newtheorem{dfn}{Definition}
\newcommand{\qedsymbol}{\vrule width.6em height.5em depth.1em\relax}
\newcommand{\qed}{\mbox{}\hfill\qedsymbol}
\newcommand{\argmin}{\mbox{arg\,min}}

\title{\Huge Multirate Anypath Routing in Wireless Mesh Networks}

\author{\IEEEauthorblockN{Rafael Laufer and Leonard Kleinrock}
 \IEEEauthorblockA{Computer Science Department\\
 University of California at Los Angeles}\\

August 29, 2008\\
Technical Report UCLA-CSD-TR080025
}

\maketitle

\begin{abstract}
In this paper, we present a new routing paradigm that generalizes opportunistic routing in wireless mesh networks. In multirate anypath routing, each node uses both a set of next hops and a selected transmission rate to reach a destination. Using this rate, a packet is broadcast to the nodes in the set and one of them forwards the packet on to the destination. To date, there is no theory capable of jointly optimizing both the set of next hops and the transmission rate used by each node. We bridge this gap by introducing a polynomial-time algorithm to this problem and provide the proof of its optimality. The proposed algorithm runs in the same running time as regular shortest-path algorithms and is therefore suitable for deployment in link-state routing protocols. We conducted experiments in a 802.11b testbed network, and our results show that multirate anypath routing performs on average 80\% and up to 6.4 times better than anypath routing with a fixed rate of 11~Mbps. If the rate is fixed at 1~Mbps instead, performance improves by up to one order of magnitude.
\end{abstract}

\section{Introduction}

The high loss rate and dynamic quality of links make routing in wireless mesh networks extremely challenging~\cite{aguayo04,campista08}. Anypath routing\footnote{We use the term anypath rather than opportunistic routing, since opportunistic routing is an overloaded term also used for opportunistic contacts~\cite{seth07}.} has been recently proposed as a way to circumvent these shortcomings by using multiple next hops for each destination~\cite{chachulski07b,biswas05a,zhong06a,dubois-ferriere07}. Each packet is broadcast to a forwarding set composed of several neighbors, and the packet must be retransmitted only if none of the neighbors in the set receive~it. Therefore, while the link to a given neighbor is down or performing poorly, another nearby neighbor may receive the packet and forward it on. This is in contrast to single-path routing where only one neighbor is assigned as the next hop for each destination. In this case, if the link to this neighbor is not performing well, a packet may be lost even though other neighbors may have overheard it. 

Existing work on anypath routing has focused on wireless networks that use a single transmission rate. This approach, albeit straightforward, presents two major drawbacks. First, using a single rate over the entire network underutilizes available bandwidth resources. Some links may perform well at a higher rate, while others may only work at a lower rate. Secondly and most importantly, the network may become disconnected at a higher bit rate. We provide experimental measurements from a 802.11b testbed which show that this phenomenon is not uncommon in practice. The key problem is that higher transmission rates have a shorter radio range, which reduces network density and connectivity. As the bit rate increases, links becomes lossier and the network eventually gets disconnected. Therefore, in order to guarantee connectivity, single-rate anypath routing must be limited to low rates.

In {\it multirate anypath routing}, these problems do not exist; however, we face additional challenges. First, we must find not only the forwarding set, but also the transmission rate at each hop that jointly minimizes its cost to a destination. Secondly, loss probabilities usually increase with higher transmission rates, so a higher bit rate does not always improve throughput. Finally, higher rates have a shorter radio range and therefore we have a different connectivity graph for each rate. Lower rates have more neighbors available for inclusion in the forwarding set (i.e., more spatial diversity) and less hops between nodes. Higher rates have less spatial diversity and longer routes. Finding the optimal operation point in this tradeoff is the focus of this paper. 

We thus address the problem of finding both a forwarding set and a transmission rate for every node, such that the overall cost of every node to a particular destination is minimized. We call this the {\it shortest multirate anypath problem}. To our knowledge, this is still an open problem~\cite{biswas05a,chachulski07b,zeng08} and we believe our algorithm is the first solution for it.

We introduce a polynomial-time algorithm to the shortest multirate anypath problem and present a proof of its optimality. Our solution generalizes Dijkstra's algorithm for the multirate anypath case and is applicable to link-state routing protocols. One would expect that the running time of such an algorithm is longer than a shortest-path algorithm. However, we show that it has the {\it same} running time as the corresponding shortest-path algorithm, being suitable for implementation at current wireless routers. We also introduce a novel routing metric that generalizes the expected transmission time (ETT) metric~\cite{draves04b} for multirate anypath routing. 

For the performance evaluation, we conducted experiments in an 18-node 802.11b wireless testbed of embedded Linux devices. Our results reveal that the network becomes disconnected if we fix the transmission rate at 2, 5.5, or 11~Mpbs. A single-rate routing scheme therefore performs poorly in this case, since 1~Mbps is the only rate at which the network is fully connected. We show that multirate anypath routing improves the end-to-end expected transmission time by 80\% on average and by up to 6.4 times compared to single-rate anypath routing at 11~Mbps, while still maintaining network connectivity. The performance is even higher over the single-rate case at 1~Mbps, with an average gain of a factor of 5.4 and a maximum gain of a factor of 11.3.

The remainder of paper is organized as follows. Section~\ref{sec:anypath} reviews the basic concepts of anypath routing and our network model. In Section~\ref{sec:multirate}, we introduce multirate anypath routing and the proposed routing metric. Section~\ref{sec:shortest} presents the multirate anypath algorithm and proves its optimality. Section~\ref{sec:results} reveals our experimental results, showing the benefits of multirate over single-rate anypath routing. Section~\ref{sec:related} presents the related work in anypath routing. Finally, conclusions are presented in Section~\ref{sec:conclusions}.

\section{Anypath Routing}
\label{sec:anypath}

In this section we review the theory of anypath routing introduced by Zhong~{\it et~al.}~\cite{zhong06a} and Dubois-Ferrière~{\it et~al.}~\cite{dubois-ferriere07}.
The main contributions of the paper are presented later in Sections~\ref{sec:multirate} and~\ref{sec:shortest}.

\subsection{Overview}

In classic wireless network routing, each node forwards a packet to a single next hop. As a result, if the transmission to that next hop fails, the node needs to retransmit the packet even though other neighbors may have overheard it. In contrast, in anypath routing, each node broadcasts a packet to {\it multiple} next hops simultaneously. Therefore, if the transmission to one neighbor fails, an alternative neighbor who received the packet can forward it on. We define this set of multiple next hops as the {\it forwarding set} and we usually use~$J$ to represent it throughout the paper. A different forwarding set is used to reach each destination, in the same way a distinct next hop is used for each destination in classic routing. 

When a packet is broadcast to the forwarding set, more than one node may receive the same packet. To avoid unnecessary duplicate forwarding, only one of these nodes should forward the packet on. For this purpose, each node in the set has a priority in relaying the received packet. A node only forwards a packet if all higher priority nodes in the set failed to do so. Higher priorities are assigned to nodes with shorter distances to the destination. As a result, if the node with the shortest distance in the forwarding set successfully received the packet, it forwards the packet to the destination while others suppress their transmission. Otherwise, the node with the second shortest distance forwards the packet, and so~on. A reliable anycast scheme~\cite{jain08} is necessary to enforce this relay priority. We talk more about this in Section~\ref{sec:model_assumptions}. The source keeps rebroadcasting the packet until someone in the forwarding set receives it or a threshold is reached. Once a neighbor in the set receives the packet, this neighbor repeats the same procedure until the packet is delivered to the destination.
Since we now use a set of next hops to forward packets, every two nodes will be connected through a mesh composed of the union of multiple paths. Figure~\ref{fig:anypath} depicts this scenario where each node uses a set of neighbors to forward packets. The forwarding sets are defined by the multiple bold arrows leaving each node. We define this union of paths between two nodes as an {\it anypath}. In the figure, the anypath shown in bold is composed by the union of 11 different paths between a source~$s$ and a destination~$d$. Depending on the choice of each forwarding set, different paths are included in or excluded from the anypath. At every hop, only a single node of the set forwards the packet on. Consequently, every packet from~$s$ traverses only one of the available paths to reach~$d$. We show a path possibly taken by a packet using a dashed line. Succeeding packets, however, may take completely different paths; hence the name anypath. The path taken is determined on-the-fly, depending on which nodes of the forwarding sets successfully receive the packet at each hop.

\begin{figure}[ht!]
\centering
\includegraphics[width=.35\textwidth]{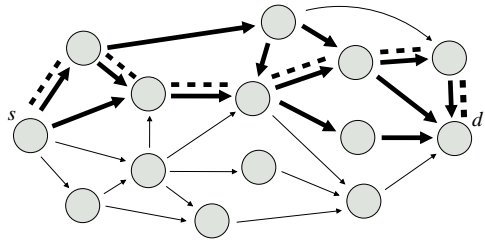}
\caption{An anypath connecting nodes $s$ and $d$ is shown in bold arrows. The anypath is composed of the union of 11 paths between the two nodes. Every packet sent from $s$ traverses one of these paths to reach $d$, such as the path shown with a dashed line. Different packets may traverse different paths, depending on which nodes receive the forwarded packet at each hop; hence the name anypath.}
\label{fig:anypath}
\end{figure}

\subsection{System Model and Assumptions}
\label{sec:model_assumptions}

In order to support the point-to-multipoint links used in anypath routing, we model the wireless mesh network as a hypergraph. A hypergraph $\mathcal{G}=(V,\mathcal{E})$ is composed of a set~$V$  of vertices or nodes and a set~$\mathcal{E}$ of hyperedges or hyperlinks. A hyperlink is an ordered pair $(i,J)$, where $i \in V$ is a node and~$J$ is a nonempty subset of~$V$ composed of neighbors of $i$. For each hyperlink $(i,J) \in \mathcal{E}$, we have a delivery probability~$p_{iJ}$ and a distance~$d_{iJ}$. If the set $J$ has a single element~$j$, then we just use $j$ instead of $J$ in our notation. In this case, $p_{ij}$ and $d_{ij}$ denote the link delivery probability and distance, respectively. 

The hyperlink delivery probability $p_{iJ}$ is defined as the probability that a packet transmitted from~$i$ is successfully received by at least one of the nodes in~$J$. One would expect that the receipt of a packet at each neighbor is correlated due to noise and interference. However, we conducted experiments which suggest that the loss of a packet at different receivers occur independently in practice, which is also consistent with other studies~\cite{reis06,miu05}. We show these results in Section~\ref{sec:results}. With the assumption of independent losses, $p_{iJ}$ is
\begin{equation}
\label{eq:p_iJ}
p_{iJ} = 1-\prod_{j \in J} \left(1-p_{ij}\right).
\end{equation}

Previously proposed MAC protocols have been designed to guarantee the relay priority among the nodes in the forwarding set~\cite{biswas05a,zorzi03,jain08}. Such protocols can use different strategies for this purpose, such as time-slotted access, prioritized contention and frame overhearing. Reliable anycast is an active area of research~\cite{jain08} and we assume that such mechanism is in place to make sure that the relaying priority is respected. The details of the MAC, however, are abstracted from the routing layer. Practical routing protocols only incorporate the delivery probabilities into the routing metric in order to abstract from the MAC details~\cite{couto03,draves04b} and we take the same approach. The only MAC aspect that is important is the effectiveness of the relaying node selection. As long as the relaying node is actually the one with the shortest distance to the destination, there should be no significant impact on the routing performance.

\subsection{Anypath Cost}
\label{sec:anypath_cost}

We are interested in calculating the anypath cost from a node~$i$ to a given destination via a forwarding set~$J$. The anypath cost $D_i$ is defined as $D_i = d_{iJ} + D_J$, which is composed of the hyperlink cost $d_{iJ}$ from~$i$ to~$J$ and the remaining-anypath cost $D_J$ from~$J$ to the destination.

The hyperlink cost $d_{iJ}$ depends on the routing metric used. 
Most previous works on anypath routing have adopted the expected number of anypath transmissions (EATX) as the routing metric~\cite{zhong06a, chachulski07b, dubois-ferriere07}. The EATX is a generalization of the unidirectional ETX metric~\cite{couto03}, which is defined as $d_{ij} = 1/p_{ij}$. The distance~$d_{ij}$ for ETX represents the expected number of transmissions necessary for a packet sent by~$i$ to be successfully received by~$j$. For EATX, the distance~$d_{iJ}$ is defined as $d_{iJ} = 1/p_{iJ}$, which is the average number of transmissions necessary for at least one node in~$J$ to correctly receive the transmitted packet.

The remaining-anypath cost $D_J$ is intuitively defined as a {\it weighted average} of the distances of the nodes in the forwarding set as
\begin{equation}
\label{eq:D_J_def} D_J = \sum_{j \in J} w_jD_j, \mathrm{\ with\ }\sum_{j \in J} w_j = 1,
\end{equation}
where the weight~$w_j$ in~(\ref{eq:D_J_def}) is the probability of node~$j$ being the relaying node. For example, let $J = \{1,2,\ldots,n\}$ with distances $D_1 \leq D_2 \leq \ldots \leq D_n$. We refer to the probability $p_{ij}$ simply by $p_j$ for convenience. Node~$j$ will be the relaying node only when it receives the packet and none of the nodes closer to the destination receives it, which happens with probability $p_j(1-p_{j-1})(1-p_{j-2})\ldots(1-p_1)$. The weight~$w_j$ is then defined as
\begin{equation}
\label{eq:w_j_def} w_j = \frac{\displaystyle p_j\prod_{k=1}^{j-1}\left(1-p_k\right)}{1-\displaystyle\prod_{j \in J} \left(1-p_j\right)},
\end{equation}
with the denominator being the normalizing constant. 

As an example, consider the network depicted in Figure~\ref{fig:example}. The distance via $J$ in Figure~\ref{fig:example-a} is calculated as 
\setlength{\arraycolsep}{0.0em}
\begin{eqnarray}
\nonumber D_i &{}={}& d_{iJ} + D_J \\
\nonumber    &{}={}& \frac{1}{1-(1-1/4)(1-1/5)} + \frac{(1/4)3 + (3/4)(1/5)3}{1-(1-1/4)(1-1/5)} \\
             &{}={}& 2.5 + 3.0 = 5.5. 
\end{eqnarray}
\setlength{\arraycolsep}{5pt}%
One would expect that adding an extra node to the forwarding set is always beneficial because it increases the number of possible paths a packet can take. However, this is not always true, as shown in Figure~\ref{fig:example-b}. The anypath distance via $J' = J \cup \{j\}$ is $D_i = d_{iJ'} + D_{J'} = 1.8 + 4.6 = 6.4$. On one hand, using~$J'$ instead of~$J$ reduces the hyperlink cost, that is, $d_{iJ'} \leq d_{iJ}$. On the other hand, the extra node increases the remaining anypath cost, that is, $D_{J'} \geq D_J$. If the increase $D_{J'}-D_J$ is higher than the decrease $d_{iJ} - d_{iJ'}$, adding this extra node is not worthy since the total cost to reach the destination increases. The intuition here is that when node~$j$ is the only one in~$J'$ that received the packet, it is cheaper to retransmit the packet to one of the two nodes in~$J$ and take a shorter path from there than to take the long path via node~$j$.

\begin{figure}[ht!]
\centering
\subfigure[]{
  \label{fig:example-a}
  \includegraphics[width=.20\textwidth]{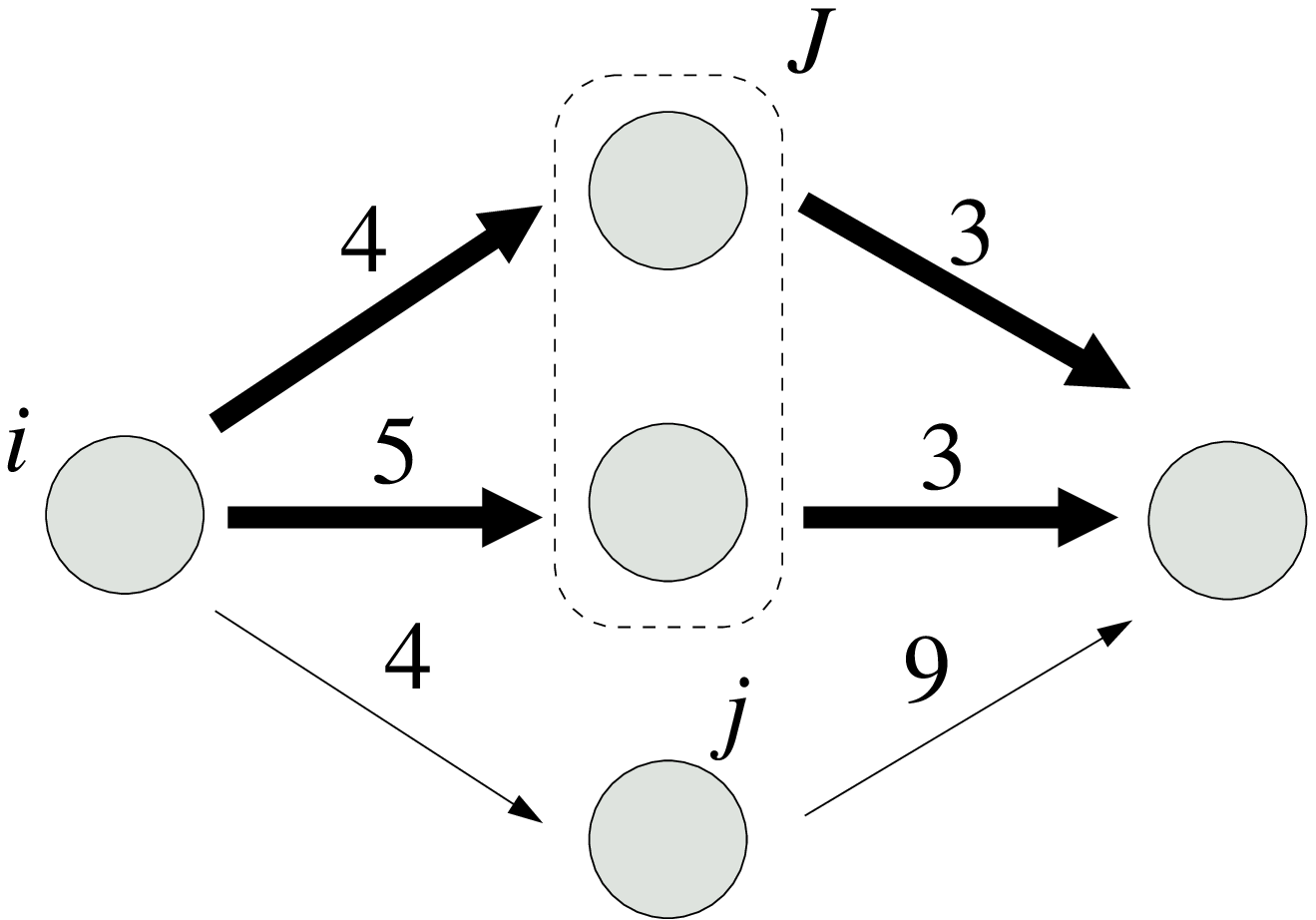}
}
\subfigure[]{
  \label{fig:example-b}
  \includegraphics[width=.20\textwidth]{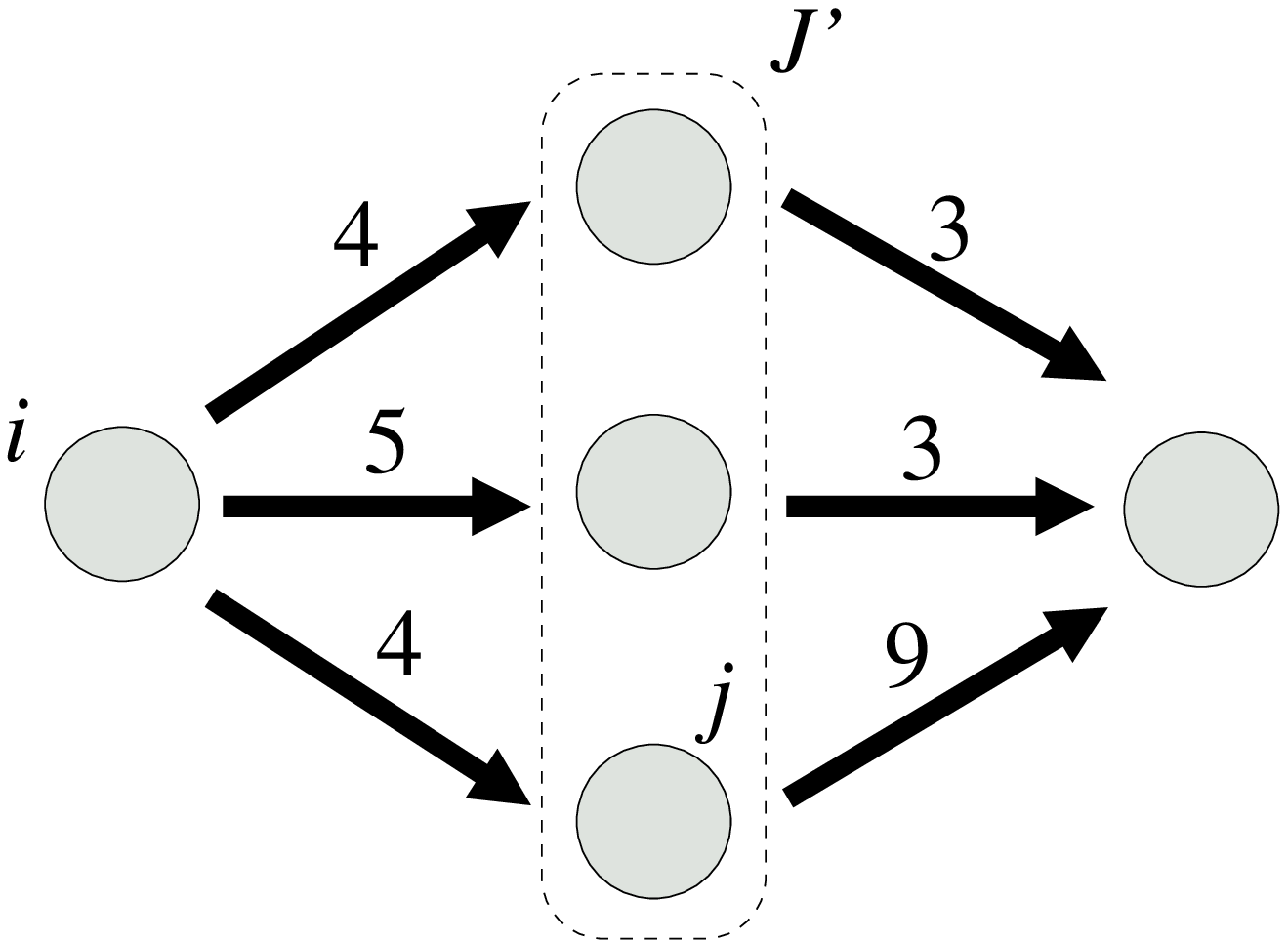}
}
\caption{An anypath cost calculation example. The weight of each link is the expected number of transmissions (ETX), which is the inverse of the link delivery probability. The anypath cost in~(a) is lower than the cost in~(b).}
\label{fig:example}
\end{figure}

Once the cost of an anypath is defined, it is of interest to find the anypath with the lowest cost to the destination, that is, the shortest anypath. This is called the {\it shortest-anypath problem}~\cite{dubois-ferriere07}. Interestingly enough, the shortest anypath will always have an equal or lower cost than the shortest single path. This is a direct consequence of the definition of an anypath as a set of paths. Among all possible anypaths between two nodes, we also have the anypath composed only of the path with the shortest ETX. Therefore, if we are to choose the shortest anypath among all these possibilities, we know for sure that its cost can never be higher than the cost of the shortest single path.

\section{Multirate Anypath Routing}
\label{sec:multirate}

Previous work on anypath routing focused on a single bit rate~\cite{biswas05a,zhong06a,chachulski07b,dubois-ferriere07}. Such an assumption, however, considerably underutilizes available bandwidth resources. Some hyperlinks may be able to sustain a higher transmission rate, while others may only work at a lower rate. To date, the problem of how to select the transmission rate for anypath routing is still open~\cite{chachulski07a,zeng08}. We provide a solution to this problem and incorporate the multirate capability inherent in IEEE 802.11 networks into anypath routing. In this case, besides selecting a set of next hops to forward packets, a node must also select one among multiple transmission rates. For each destination, a node then keeps both a forwarding set and a transmission rate used to reach this set. As a result, every two nodes will be connected through a mesh composed of the union of multiple paths, with each node transmitting at a selected rate. Figure~\ref{fig:multirate_anypath} depicts the scenario where nodes use a selected bit rate to forward packets to a set of neighbors. We define this union of paths between two nodes, with each node using a potentially different bit rate as a {\it multirate anypath}. In the figure, assume that a packet is sent from~$s$ to~$d$ over the multirate anypath. Only one of the available paths is traversed depending on which nodes successfully receive the packet at each hop. We show a path possibly taken by the packet using a dashed line. We use different dash lengths to represent the different transmission rates used by each node. A shorter dash represents a shorter time to send a packet, hence a higher transmission rate. Succeeding packets may take completely different paths with other transmission rates along its way.

\begin{figure}[ht!]
\centering
\includegraphics[width=.35\textwidth]{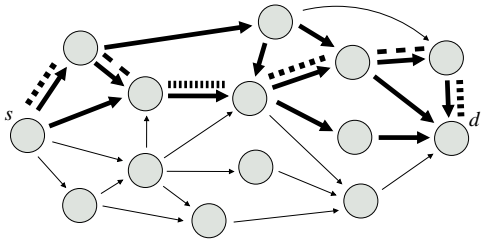}
\caption{A multirate anypath connecting nodes $s$ and $d$ is shown in bold arrows. Every packet sent from~$s$ traverses a path to reach~$d$, such as the path shown with dashed lines. Different dash lengths represent the different bit rates used by each node, with a shorter dash for higher rates.}
\label{fig:multirate_anypath}
\end{figure}

In order to support multirate, we must extend the system model in Section~\ref{sec:model_assumptions}. Let~$R$ be the set of available bit rates that nodes can use to transmit their packets. For each hyperlink $(i,J) \in \mathcal{E}$, we now have a delivery probability~$p_{iJ}^{(r)}$ and a distance~$d_{iJ}^{\,(r)}$ associated with each transmission rate $r \in R$. In real wireless networks, we usually have different delivery probabilities and distances for each transmission rate, which justifies this model extension. 

The EATX metric described in Section~\ref{sec:anypath_cost} was originally designed considering that nodes transmit at a single bit rate. To account for multiple bit rates, we introduce the expected anypath transmission time (EATT) metric. For EATT, the hyperlink distance $d_{iJ}^{\,(r)}$ for each rate $r \in R$ is defined as
\begin{equation}
\label{eq:d_iJ_ETT} d_{iJ}^{\,(r)} = \frac{1}{p_{iJ}^{(r)}} \times \frac{s}{r},
\end{equation}
where $p_{iJ}^{(r)}$ is the hyperlink delivery probability defined in~(\ref{eq:p_iJ}), $s$ is the maximum packet size, and $r$ is the bit rate. The distance $d_{iJ}^{\,(r)}$ is basically the time it takes to transmit a packet of size~$s$ at a bit rate~$r$ over a lossy hyperlink with delivery probability~$p_{iJ}^{(r)}$. The EATT metric is a generalization of the expected transmission time (ETT) metric~\cite{draves04b} commonly used in single-path wireless routing. Note that for each bit rate $r \in R$, we have a different delivery probability~$p_{iJ}^{(r)}$, which usually decreases for higher rates. This behavior imposes a tradeoff; a higher bit rate decreases the time of a single packet transmission (i.e., $s/r$ decreases), but it usually increases the number of transmissions required for a packet to be successfully received (i.e., $1/p_{iJ}^{(r)}$ increases).

The remaining-anypath cost~$D_J^{(r)}$ now also depends on the transmission rate, since the delivery probabilities change for each rate. Since both the hyperlink distance and the remaining anypath cost depend on the bit rate, node~$i$ has a different anypath cost~$D_i^{(r)} = d_{iJ}^{\,(r)} + D_J^{(r)}$ for each forwarding set~$J$ and for each transmission rate $r \in R$. The remaining-anypath cost~$D_J^{(r)}$ for a rate $r \in R$ is defined as
\begin{equation}
\label{eq:D_J_r_def} D_J^{(r)} = \sum_{j \in J} w_j^{(r)}D_j, \mathrm{\ with\ }\sum_{j \in J} w_j^{(r)} = 1,
\end{equation}
where the weight~$w_j^{(r)}$ in~(\ref{eq:D_J_r_def}) is the probability of node~$j$ being the relaying node and $D_j = \min_{r \in R} D_j^{(r)}$ is the shortest distance from node~$j$ to the destination among all rates. We refer to the probability $p_{ij}^{(r)}$ simply by $p_j^{(r)}$ for convenience. The weight~$w_j^{(r)}$ is then defined as
\begin{equation}
w_j^{(r)} = \frac{\displaystyle p_j^{(r)}\prod_{k=1}^{j-1}\left[1-p_k^{(r)}\right]}{1-\displaystyle\prod_{j \in J} \left[1-p_j^{(r)}\right]}.
\end{equation}

We address the problem of finding both the forwarding set and the transmission rate that minimize the overall cost to reach a particular destination. We call this the {\it shortest multirate anypath problem}, which generalizes the shortest-anypath problem~\cite{dubois-ferriere07} for the multirate scenario. Interestingly, the shortest multirate anypath will always have equal or lower cost than the shortest single path. Among all possible multirate anypaths between two nodes, we also have the single path with the shortest ETT. As a result, the cost of the shortest multirate anypath can never be higher than the cost of the shortest path. Likewise, due to the same argument, the shortest multirate anypath will also have equal or lower cost than any shortest anypath using a single transmission rate.

\section{Finding the Shortest Multirate Anypath}
\label{sec:shortest}

In this section we introduce the proposed shortest-anypath algorithms. In Section~\ref{sec:single-rate}, we present the Shortest Anypath First (SAF) algorithm used in a single-rate network with the EATX metric. A similar single-rate algorithm along with the same optimization was also proposed by Chachulski~\cite{chachulski07b}. We, however, derived this algorithm independently and later in Section~\ref{sec:multi-rate} we introduce a generalization of this algorithm for multiple rates. Surprisingly, the Shortest Multirate Anypath First (SMAF) algorithm has the {\it same} running time as a shortest single-path algorithm for multirate.
We only show the proof of optimality of the SMAF algorithm, since by definition this implies the optimality of the SAF algorithm.

\subsection{The Single-Rate Case}
\label{sec:single-rate}

\begin{figure*}[ht!]
\centering
\subfigure[]{
  \label{fig:saf-execution-a}
  \includegraphics[width=.23\textwidth]{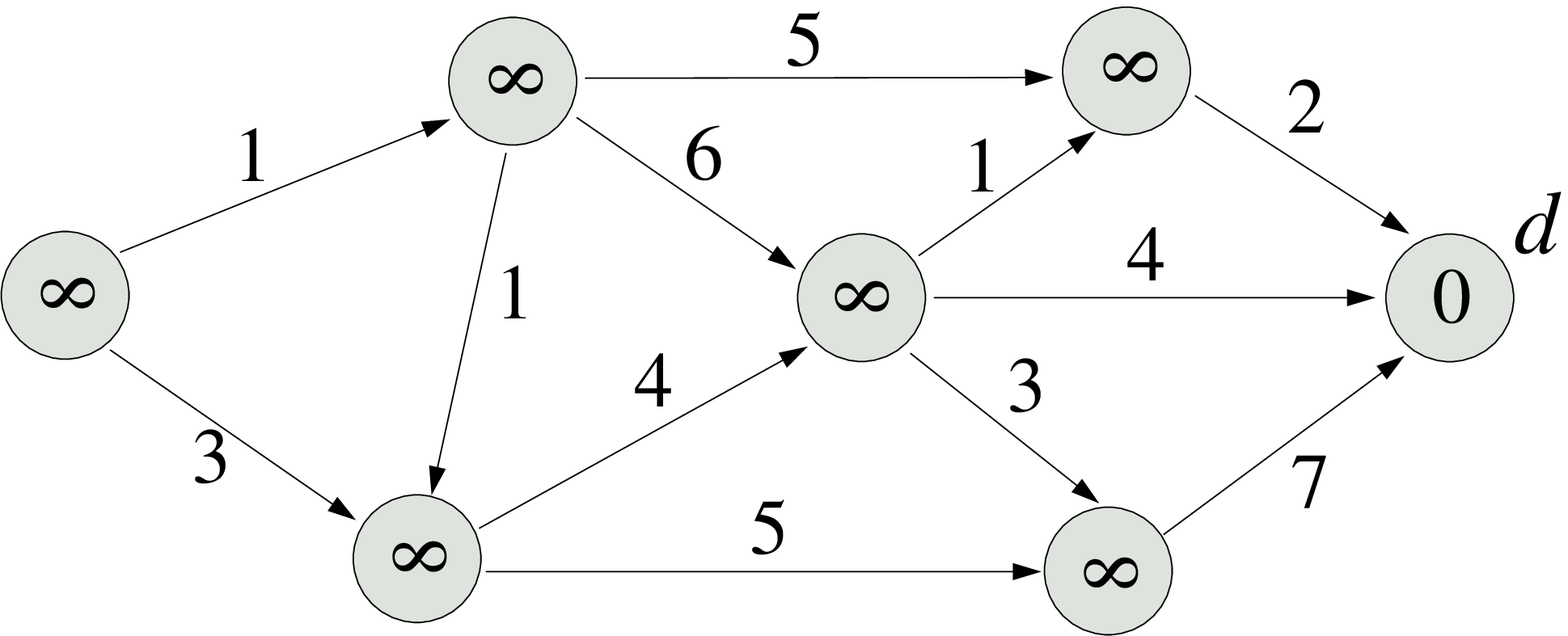}
}
\subfigure[]{
  \label{fig:saf-execution-b}
  \includegraphics[width=.23\textwidth]{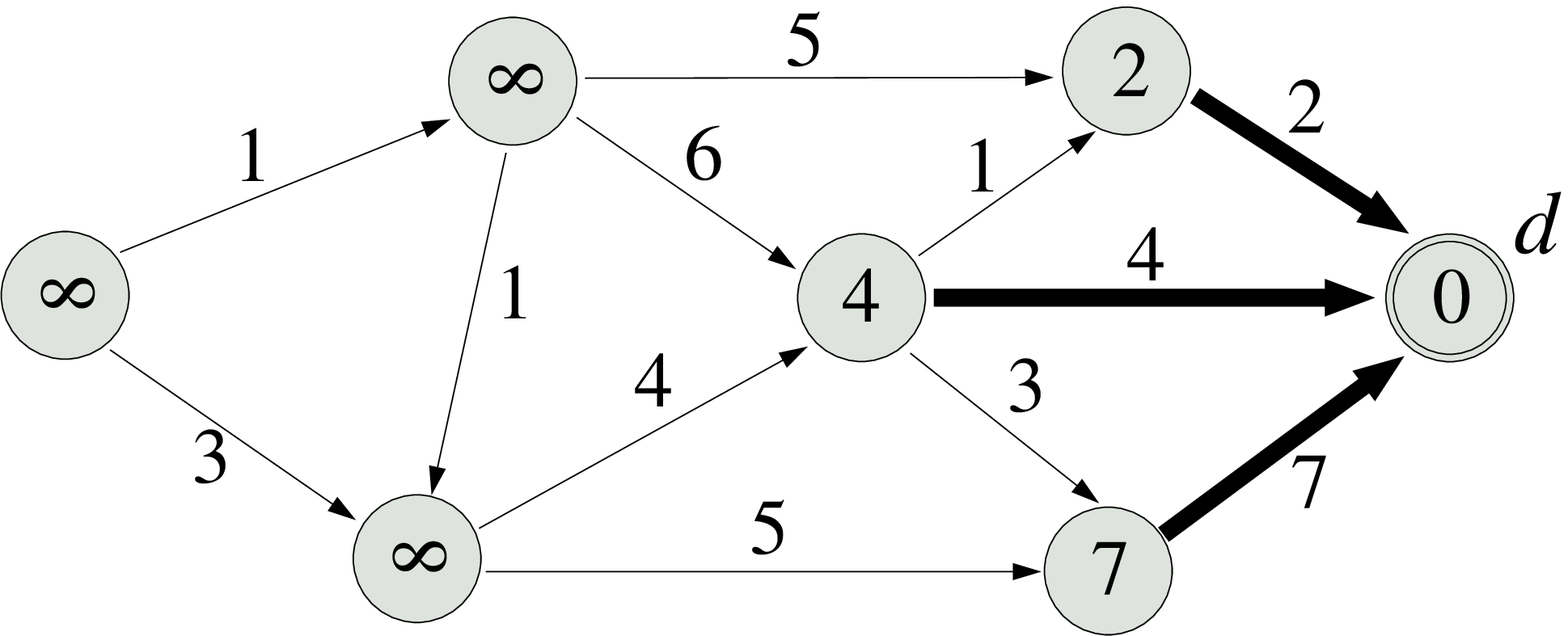}
}
\subfigure[]{
  \label{fig:saf-execution-c}
  \includegraphics[width=.23\textwidth]{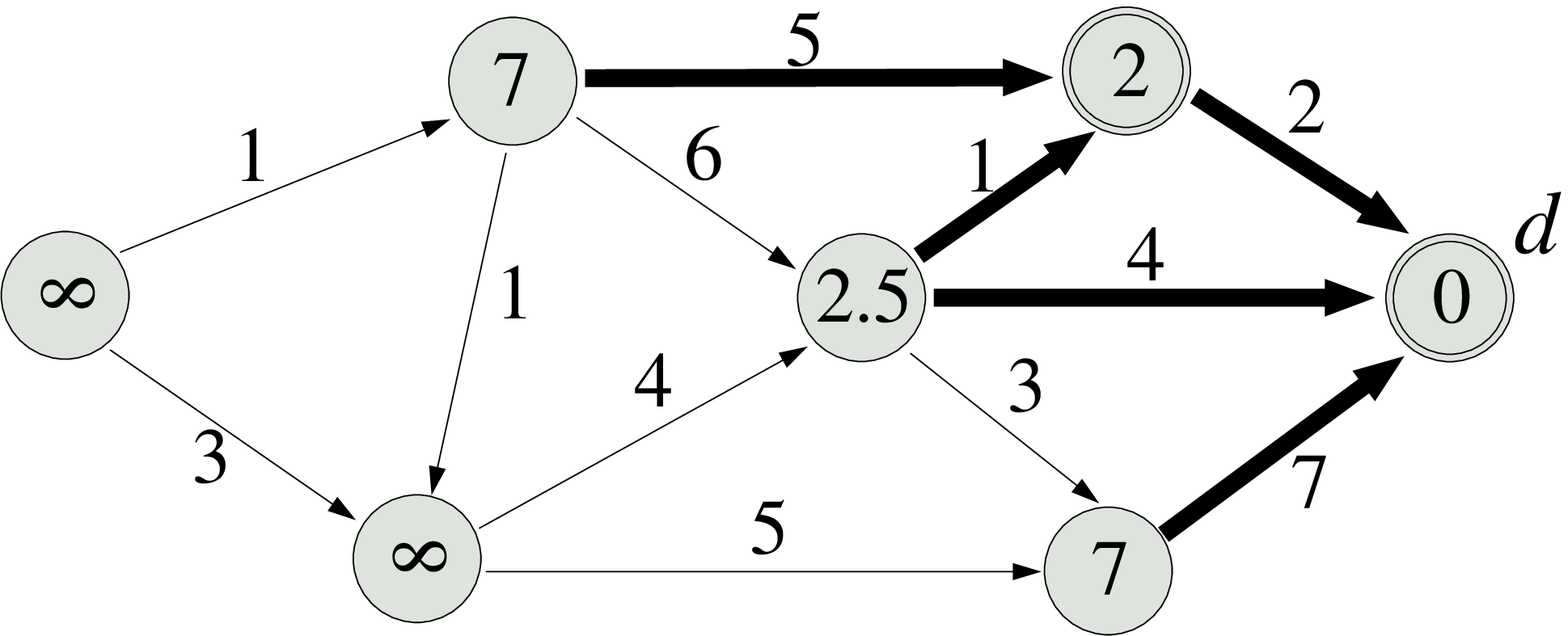}
}
\subfigure[]{
  \label{fig:saf-execution-d}
  \includegraphics[width=.23\textwidth]{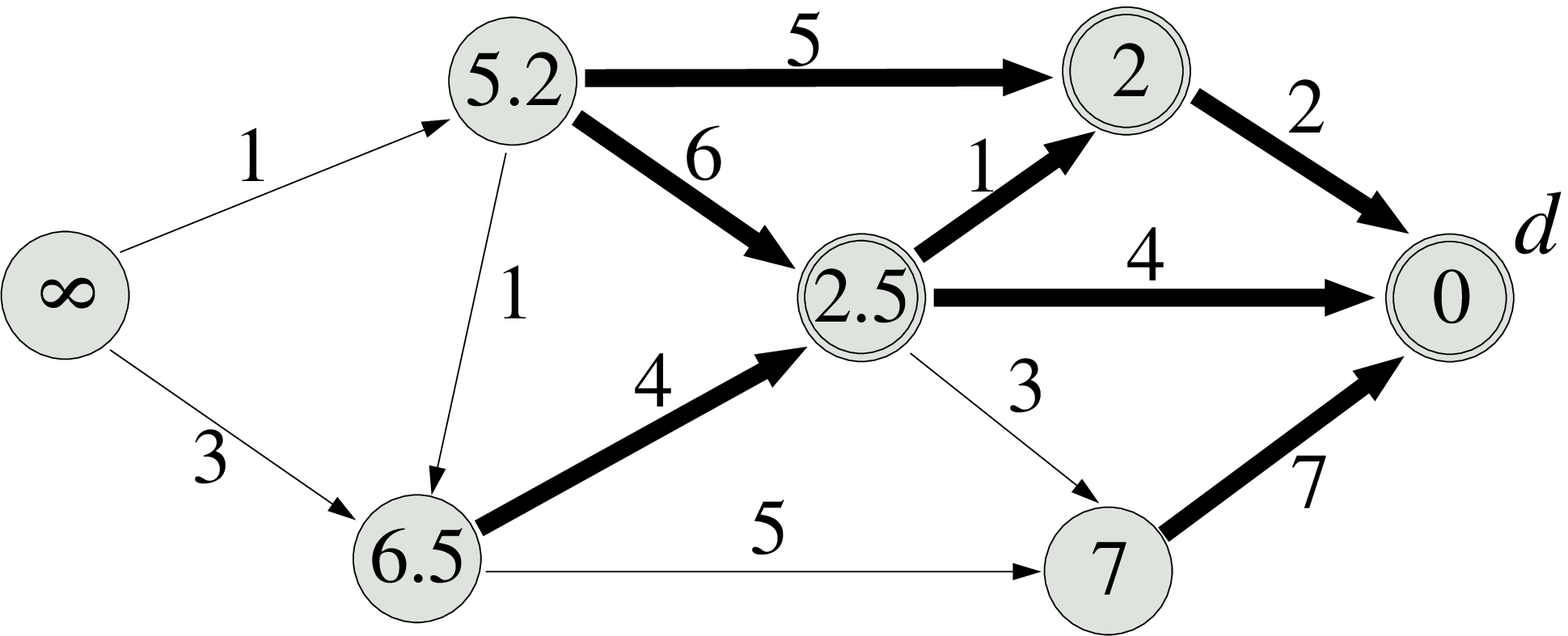}
}
\subfigure[]{
  \label{fig:saf-execution-e}
  \includegraphics[width=.23\textwidth]{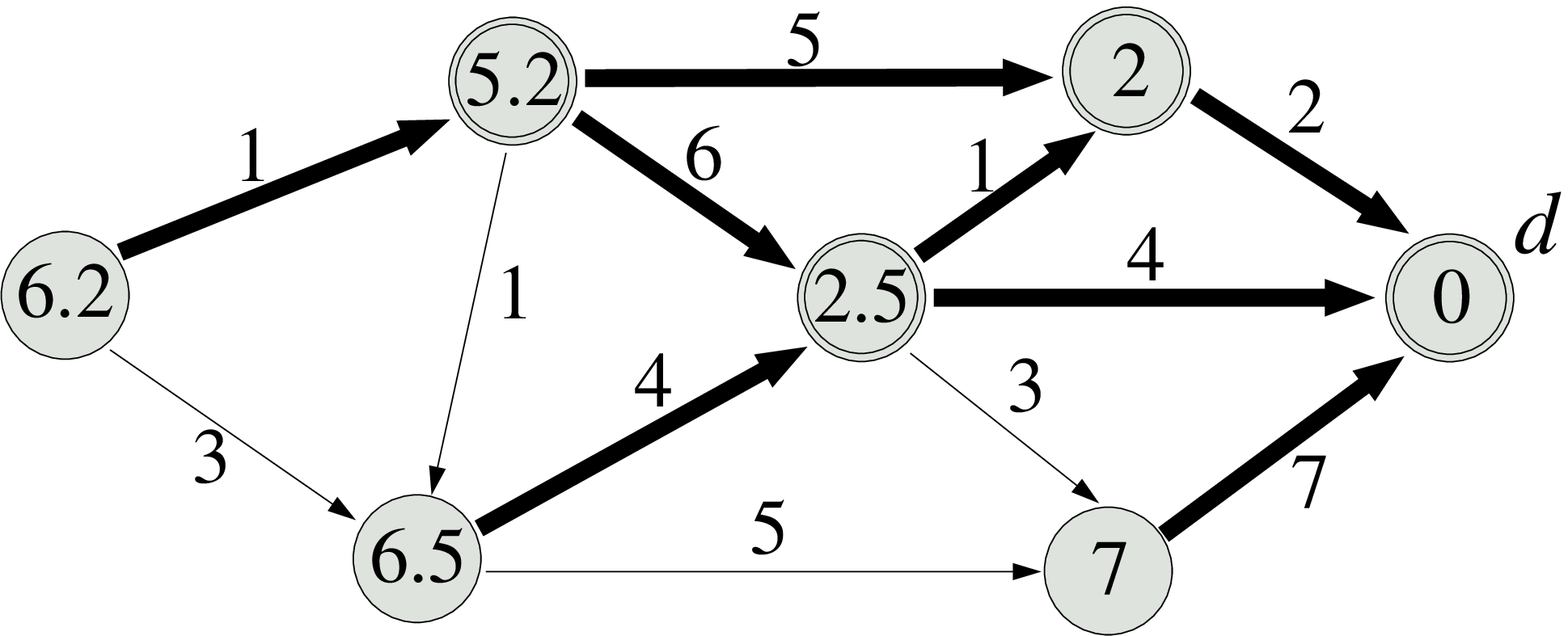}
}
\subfigure[]{
  \label{fig:saf-execution-f}
  \includegraphics[width=.23\textwidth]{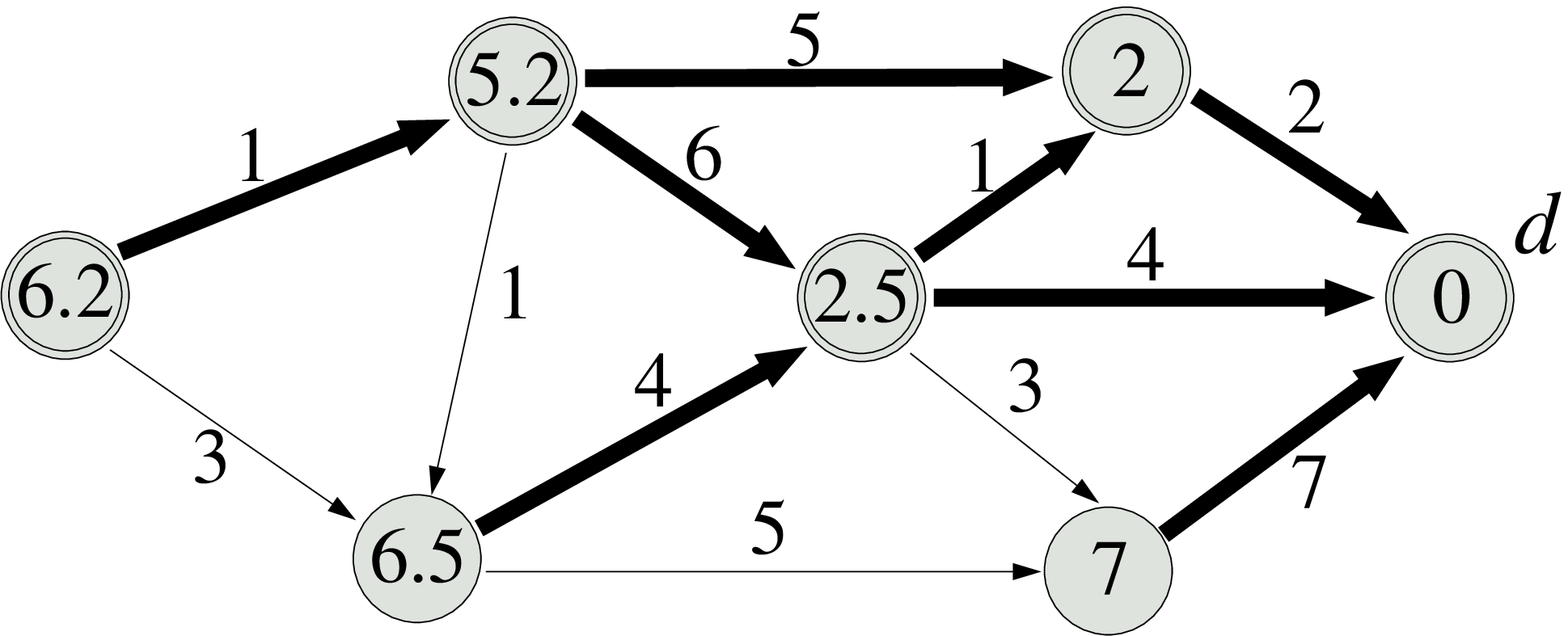}
}
\subfigure[]{
  \label{fig:saf-execution-g}
  \includegraphics[width=.23\textwidth]{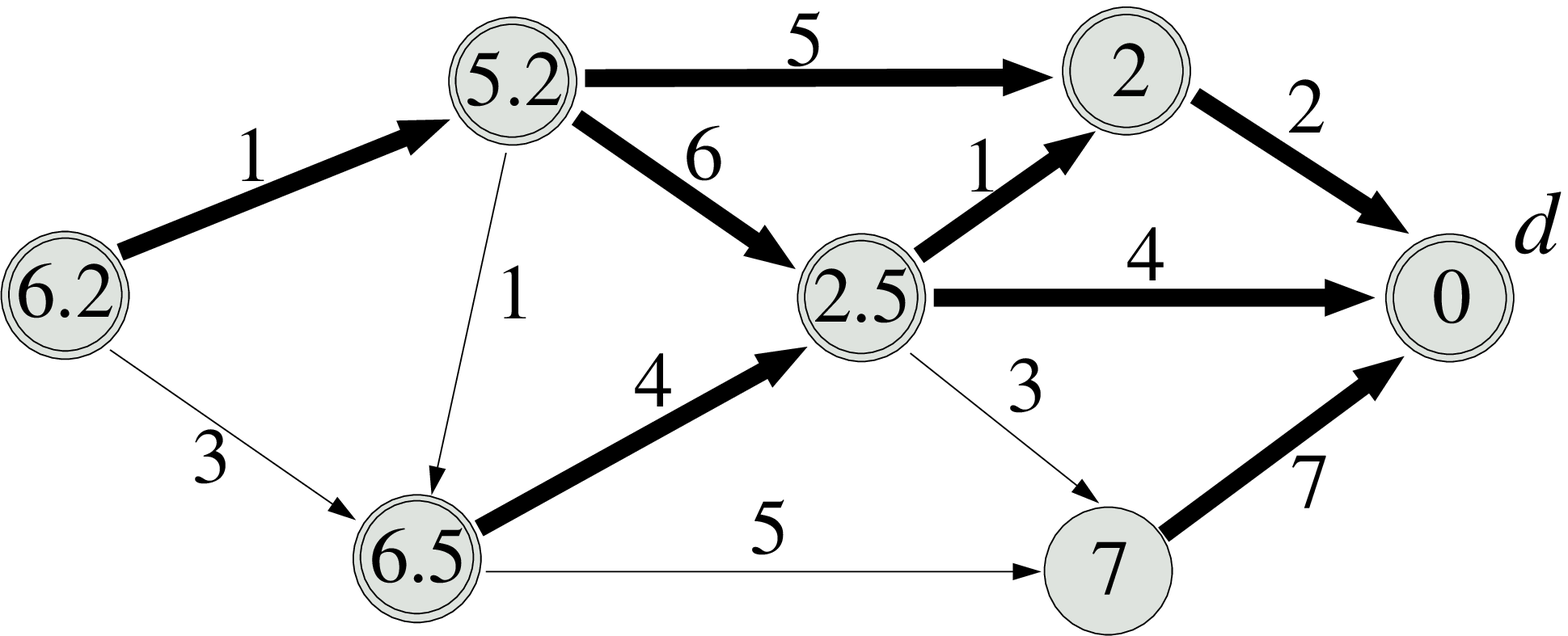}
}
\subfigure[]{
  \label{fig:saf-execution-h}
  \includegraphics[width=.23\textwidth]{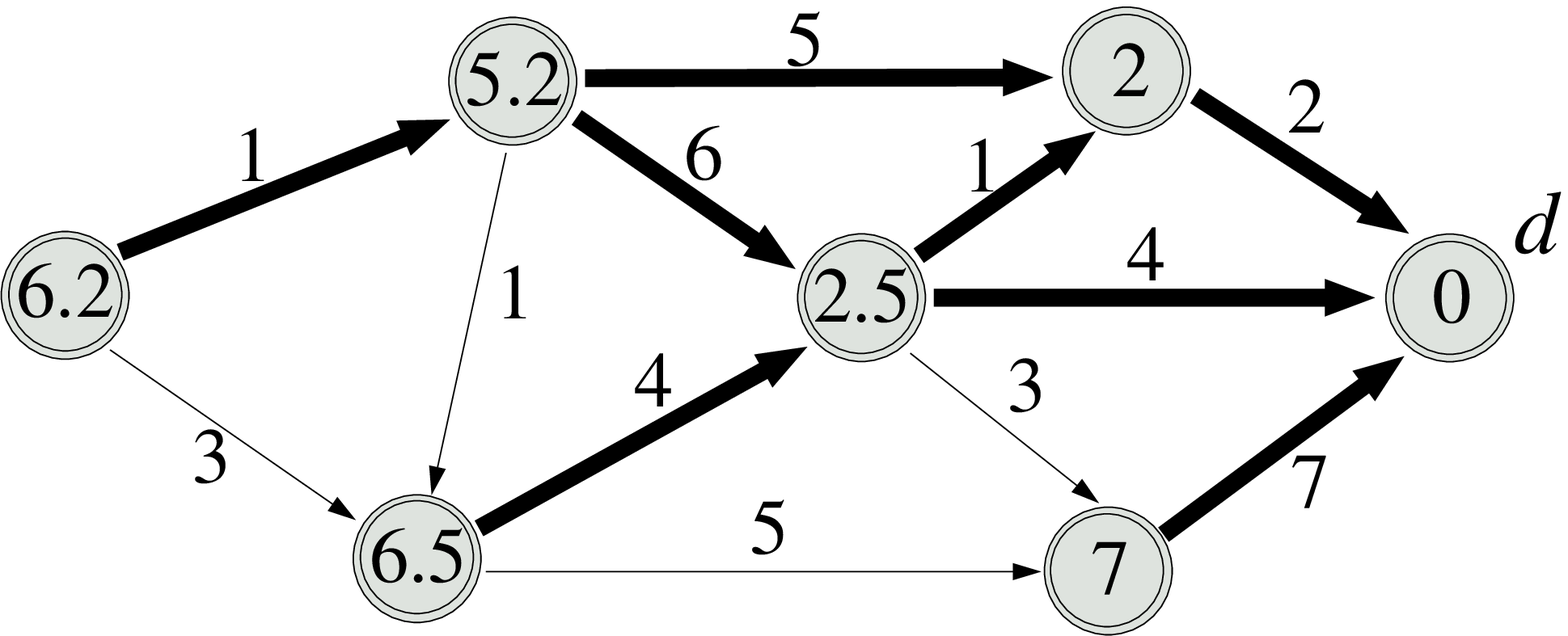}
}

\caption{Execution of the Shortest Anypath First (SAF) algorithm from every node to~$d$. The weight of each link is the expected number of transmissions (ETX), which is the inverse of the link delivery probability. (a) The situation just after the initialization. (b)--(h) The situation after each successive iteration of the algorithm. Part (h) shows the situation after the last node is settled.}
\label{fig:saf-execution}
\end{figure*}

We now present the Shortest Anypath First algorithm used in the simpler single-rate scenario. Given a graph $G=(V,E)$, the algorithm calculates the shortest anypaths from all nodes to a destination~$d$. For every node $i \in V$ we keep an estimate~$D_i$, which is an upper-bound on the distance of the shortest anypath from~$i$ to~$d$. In addition, we also keep a forwarding set $F_i$ for every node, which stores the set of nodes used as the next hops to reach~$d$. Finally, we keep two data structures, namely~$S$ and~$Q$. The $S$ set stores the set of nodes for which we already have a shortest anypath defined. 
We store each node $i \in V - S$ for which we still do not have a shortest anypath in a priority queue~$Q$ keyed by their $D_i$ values.
\\
\begin{algorithm}{Shortest-Anypath-First}{G,d}
\begin{FOR}{\EACH \text{\bf node} i \IN V}
D_i \= \infty \\
F_i \= \emptyset
\end{FOR} \\
D_d \= 0 \\
S \= \emptyset \\
Q \= V \\
\begin{WHILE}{Q \neq \emptyset}
j \= \CALL{Extract-Min}(Q)\\
S \= S \cup \{j\} \\
\begin{FOR}{\EACH \text{\bf incoming edge} (i,j) \IN E}
J \= F_i \cup \{j\} \\
\begin{IF}{D_i > D_j}
D_i \= d_{iJ} + D_J \\
F_i \= J
\end{IF}
\end{FOR} 
\end{WHILE}
\end{algorithm}

Lines~1--3 initialize the state variables $D_i$ and $F_i$ and line~4 sets to zero the distance from node $d$ to itself. Lines~5--6 initialize the $S$ and $Q$ data structures. Initially, we do not have the shortest anypath from any node, so $S$ is initially empty and thus $Q$ contains all the vertices in the graph. 
As in the shortest-path algorithm, the Shortest Anypath First algorithm is composed of $|V|$ rounds, dictated by the number of elements initially in~$Q$. At each round, the {\sc Extract-Min} procedure extracts the node with the minimum distance to the destination from~$Q$. Let this node be $j$. At this point, $j$ is settled and inserted into~$S$, since the shortest anypath from~$j$ to the destination is now known. 
For each incoming edge $(i,j) \in E$, we check if the distance~$D_i$ is larger than the distance $D_j$. If that is the case, then node~$j$ is added to the forwarding set $F_i$ and the distance $D_i$ is updated.

Figure~\ref{fig:saf-execution} shows the execution of Shortest Anypath First algorithm using the EATX metric. We see in Figure~\ref{fig:saf-execution-a} the graph right after the initialization. Figures~\ref{fig:saf-execution-b}--\ref{fig:saf-execution-h} show each iteration of the algorithm. At each step, the value inside a node~$i$ presents the distance $D_i$ from that node to the destination~$d$ and the arrows in boldface present the shortest anypath to~$d$. Nodes with two circles are the settled nodes in~$S$. The graph in Figure~\ref{fig:saf-execution-h} shows the result of SAF algorithm right after settling the last node.

The running time of the Shortest Anypath First algorithm depends on how $Q$ is implemented. Assuming that we have a Fibonacci heap, the cost of each of the $|V|$ {\sc Extract-Min} operations in line~8 takes $O(\log V)$, with a total of $O(V \log V)$ aggregated time. The running time to calculate both $d_{iJ}$ and $D_J$ in line~13 depends on the size of~$J$; however, if we store additional state, it can be reduced to a constant time, as we show in the next paragraph. The {\bf for} loop of lines 10--13 takes $O(E)$ aggregated time and as a result the total complexity of the algorithm is $O(V\log V + E)$, which is the same complexity of Dijkstra's algorithm.

To reduce the running time of the calculation of the node distance~$D_i$ in line~13 to $O(1)$, we keep two additional state variables for each node~$i$, namely $\alpha_i$ and $\beta_i$. In $\alpha_i$, we store
\begin{equation}
\alpha_i \leftarrow 1 + \sum_{j \in F_i} p_jD_j\prod_{k = 1}^{j-1} \left(1-p_k\right),
\end{equation}
and in $\beta_i$ we store
\begin{equation}
\beta_i \leftarrow \prod_{j \in F_i} \left(1-p_j\right).
\end{equation}
Suppose now that we must update the forwarding set~$F_i$ to include a new node~$n$, that is, $F_i \leftarrow F_i \cup \{n\}$. First, we update the state variables $\alpha_i$ and $\beta_i$ to
\setlength{\arraycolsep}{0.0em}
\begin{eqnarray}
\nonumber \alpha_i &{}\leftarrow{}& \alpha_i + \beta_i\,p_nD_n \\
          \beta_i  &{}\leftarrow{}& \beta_i\,(1-p_n),
\end{eqnarray}
\setlength{\arraycolsep}{5pt}%
and finally we update~$D_i$ with
\begin{equation}
\label{eq:D_J_update} D_i \leftarrow \frac{\alpha_i}{1-\beta_i}.
\end{equation}

\vspace{.1cm}
\subsection{The Multirate Case}
\label{sec:multi-rate}

We now generalize the SAF algorithm to support multiple transmission rates, introducing the Shortest Multirate Anypath First (SMAF) algorithm. 
For each node $i \in V$, we now keep a different distance estimate~$D_i^{(r)}$ for every rate $r \in R$. The estimate~$D_i^{(r)}$ is an upper-bound on the distance of the shortest anypath from~$i$ to~$d$ using transmission rate~$r$. In addition, we also keep its corresponding forwarding set~$F_i^{(r)}$, which stores the set of next hops used for~$i$ to reach~$d$ using~$r$. We use~$D_i$ and~$F_i$ without the indicated rates to store the minimum distance estimate among all rates and its corresponding forwarding set, respectively. We also keep a transmission rate $T_i$ for every node, which stores the optimal rate used to reach~$d$. 
\\
\begin{algorithm}{Shortest-Multirate-Anypath-First}{G,d}
\begin{FOR}{\EACH \text{\bf node} i \IN V}
D_i \= \infty \\
F_i \= \emptyset \\
T_i \= \NIL \\
\begin{FOR}{\EACH \text{\bf rate} r \IN R}
D_i^{\,(r)} \= \infty \\
F_i^{\,(r)} \= \emptyset
\end{FOR}
\end{FOR} \\
D_d \= 0 \\
S \= \emptyset \\
Q \= V \\
\begin{WHILE}{Q \neq \emptyset}
j \= \CALL{Extract-Min}(Q)\\
S \= S \cup \{j\} \\
\begin{FOR}{\EACH \text{\bf incoming edge} (i,j) \IN E}
\begin{FOR}{\EACH \text{\bf rate} r \IN R}
J \= F_i^{(r)} \cup \{j\} \\
\begin{IF}{D_i^{(r)} > D_j}
D_i^{\,(r)} \= d_{iJ}^{\,(r)} + D_J^{(r)} \\
F_i^{\,(r)} \= J \\
\begin{IF}{D_i > D_i^{(r)}}
D_i \= D_i^{(r)} \\
F_i \= F_i^{(r)} \\
T_i \= r
\end{IF}
\end{IF}
\end{FOR} 
\end{FOR}
\end{WHILE}
\end{algorithm}

The key idea of the SMAF algorithm is that each node $i \in V$ has an independent distance estimate~$D_i^{(r)}$ for each rate $r \in R$ and we keep the minimum of these estimates as the node distance~$D_i$. At each round of the {\bf while} loop, the node with the minimum distance from~$Q$ is settled. Let this node be~$j$. For each incoming edge $(i,j) \in E$, we check for every rate $r \in R$ if the distance $D_i^{(r)}$ is larger than the distance $D_j$ of the node just settled. If that is the case, then node~$j$ is added to the forwarding set $F_i^{(r)}$ of that specific rate and distance $D_i^{(r)}$ is updated accordingly. If the new distance~$D_i^{(r)}$ is shorter than the node distance~$D_i$, we update the node distance~$D_i$ as well as the forwarding set~$F_i$ and transmission rate~$T_i$ to reflect the new minimum.

The running time of the Shortest Multirate Anypath First algorithm also depends on the implementation of~$Q$. The initialization in lines~1--10 takes $O(VR)$ time. Assuming that we have a Fibonacci heap, the {\sc Extract-Min} operations in line~12 take a total of $O(V \log V)$ aggregated time. We assume that the distance calculation of~$d_{iJ}^{\,(r)}$ and~$D_J^{(r)}$ in line~18 is optimized to take a constant time, as shown in Section~\ref{sec:single-rate}. As a result, the {\bf for} loop in lines 15--23 takes $O(ER)$ aggregated time. The total running time is therefore $O(V \log V +(E+V)R)$, which is $O(V \log V + ER)$ if all nodes are able to reach the destination. This is the same running time of the shortest single-path algorithm for multiple rates. Compared to the SAF algorithm, the SMAF algorithm allows nodes to take advantage of their multiple transmission rates at the cost of just a small increase in the running time.

In order to prove the optimality of the algorithm, we first introduce five lemmas that show a few properties of multirate anypath routing. We use~$\delta_i^{(r)}$ as the distance of the shortest multirate anypath from a node~$i$ to the destination~$d$, when~$i$ transmits at a fixed rate $r \in R$. Likewise, $\phi_i^{(r)}$ represents the corresponding forwarding set used in this multirate anypath. We use~$\delta_i$ without the indicated rate to represent the distance of the shortest multirate anypath from~$i$ to~$d$ via the optimal forwarding set~$\phi_i$ and optimal transmission rate $\rho \in R$. That~is, $\delta_i = \min_{r \in R} \delta_i^{(r)}$, $\rho = \argmin_{r \in R} \delta_i^{(r)}$, and $\phi_i = \phi_i^{(\rho)}$. We use $D_i$ as the distance of a particular multirate anypath from~$i$ to~$d$, but not necessarily the shortest one. 
The proof for each of these lemmas is available in Appendix~\ref{app:proofs}. 

\begin{lem}
\label{lem:comparison} 
{\it For a fixed transmission rate, let $D_i$ be the distance of a node~$i$ via forwarding set~$J$ and let $D_i'$ be the distance via forwarding set $J' = J \cup \{n\}$, where $D_n \geq D_j$ for every node $j \in J$. We have $D_i' \leq D_i$ if and only if $D_i \geq D_n$.}
\end{lem}

We use Lemma~\ref{lem:comparison} for the comparisons in line~12 of the SAF algorithm and in line~17 of the SMAF algorithm. By this lemma, if the distance~$D_i$ via~$J$ is larger than the distance $D_n$ of a neighbor node~$n$, with $D_n \geq D_j$ for all $j \in J$, then the distance~$D_i'$ via $J'= J \cup \{n\}$ is always smaller than $D_i$. That is, it is always beneficial to include node~$n$ in the forwarding set in order to obtain a shorter distance to the destination.

\begin{lem}
\label{lem:acyclic} 
{\it The shortest distance~$\delta_i$ of a node~$i$ is always larger than or equal to the shortest distance $\delta_j$ of any node~$j$ in the optimal forwarding set $\phi_i$. That is, we have $\delta_i \geq \delta_j$ for all $j \in \phi_i$.}
\end{lem}

Lemma~\ref{lem:acyclic} guarantees that if a node~$i$ uses another node~$j$ in its optimal forwarding set~$\phi_i$, then distance~$\delta_i$ can never be smaller than~$\delta_j$. This is equivalent to the restriction that all weights in the graph must be nonnegative in Dijkstra's algorithm. 

\begin{lem}
\label{lem:extra} 
{\it For any transmission rate, if a node~$i$ uses a node~$n$ in its optimal forwarding set~$\phi_i$ and $\delta_i = \delta_n$, we can safely remove~$n$ from~$\phi_i$ without changing~$\delta_i$. The link $(i,n)$ is said to be ``redundant.''}
\end{lem}

By Lemma~\ref{lem:extra}, if the distances $\delta_i = \delta_n$ of two nodes~$i$ and~$n$ are the same, then the distance~$\delta_i$ via forwarding set~$\phi_i$ is the same as the distance via forwarding set~$\phi_i - \{n\}$. That is, the distance of node~$i$ does not change if it uses~$n$ in its forwarding set or not.

\begin{lem}
\label{lem:order} 
{\it If the shortest distances from the neighbors of a node~$i$ to a given destination are $\delta_1 \leq \delta_2 \leq \ldots \leq \delta_n$, then~$\phi_i^{(r)}$ is always of the form $\phi_i^{(r)} = \{1,2,\ldots,k\}$, for some $k \in \{1,2,\ldots,n\}$.}
\end{lem}

According to Lemma~\ref{lem:order}, the best forwarding set~$\phi_i^{(r)}$ for transmission rate $r \in R$ is a subset of neighbors with the shortest distances to the destination. That is, given a set of neighbors with distances $\delta_1 \leq \delta_2 \leq \ldots \leq \delta_n$, the best forwarding set~$\phi_i^{(r)}$ when using rate $r \in R$ is always one of $\{1\}$, $\{1,2\}$, $\{1,2,3\},\ldots,\{1,2,\ldots,n\}$. As a result, forwarding sets with gaps between the neighbors, such as $\{2,3\}$ or $\{1,4\}$, can never yield the shortest distance to the destination. This property is the key factor that allows us to reduce the complexity of the proposed algorithms from exponential to polynomial time. For $n$ neighbors, we do not have to test every one of the $2^n-1$ possible forwarding sets. Instead, we only need to check at most $n$ forwarding sets.

\begin{lem} 
\label{lem:fset} 
{\it For a given transmission rate $r \in R$, assume that $\phi_i^{(r)} = \{1,2,\ldots,n\}$ with distances $\delta_1 \leq \delta_2 \leq \ldots \leq \delta_n$. If $D_i^j$ is the distance from node~$i$ using transmission rate~$r$ via forwarding set $\{1,2,\ldots,j\}$, for $1 \leq j \leq n$, then we always have $D_i^1 \geq D_i^2 \geq \ldots \geq D_i^n = \delta_i^{(r)}$.}
\end{lem}

Lemma~\ref{lem:fset} explains another important property necessary for the SMAF algorithm to converge. Assuming now that the best forwarding set $\phi_i^{(r)}$ for transmission rate $r \in R$ is defined as $\phi_i^{(r)} = \{1,2,\ldots,n\}$ with distances $\delta_1 \leq \delta_2 \leq \ldots \leq \delta_n$, the distance $D_i$ monotonically decreases as we use each of the forwarding sets $\{1\},\{1,2\},\{1,2,3\},\ldots,\{1,2,\ldots,j\}$.

We now present the proof of optimality of the algorithm.
\begin{thm}
\label{thm:saf_correctness}
{\it
Optimality of the algorithm.

Let $G = (V,E)$ be a weighted, directed, graph and let $d$ be the destination. After running the Shortest Multirate Anypath First algorithm on $G$, we have $D_i = \delta_i$ for all nodes $i \in V$.}
\end{thm}
\begin{proof}
This proof is similar to the proof of Dijkstra's algorithm~\cite{cormen01}. We show that for each node $s \in V$, we have $D_s = \delta_s$ at the time $s$ is added to $S$. 

For the purpose of contradiction, let $s$ be the first node added to $S$ for which $D_s \neq \delta_s$. We must have $s \neq d$ because $d$ is the first node added to $S$ and $D_d = \delta_d = 0$ at that time. Just before adding $s$ to $S$, we also have that $S$ is not empty, since $s \neq d$ and $S$ must contain at least $d$. We assume that there must be a multirate anypath from $s$ to $d$, otherwise $D_s = \delta_s = \infty$, which contradicts our initial assumption that $D_s \neq \delta_s$. If there is at least one multirate anypath, there is a shortest multirate anypath~$\alpha$ from $s$ to $d$. Let us consider a cut $(V - S, S)$ of~$\alpha$, such that we have $s \in V - S$ and $d \in S$, as shown in Figure~\ref{fig:saf_proof}. Let the set~$J$ be composed of nodes in $V - S$ that have an outgoing link to a node in~$S$. 
Likewise, let the set~$K$ be composed of nodes in~$S$ that have an incoming link from a node in~$V - S$. 

\begin{figure}[h!]
\centering
\includegraphics[width=.40\textwidth]{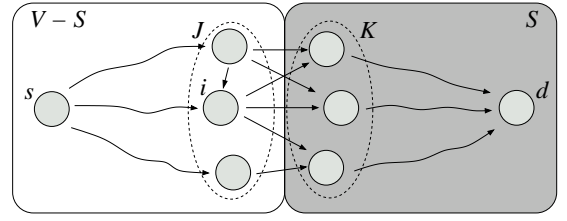}
\caption{The shortest multirate anypath~$\alpha$ from~$s$ to~$d$. Set~$S$ must be nonempty before node~$s$ is inserted into it, since it must contain at least~$d$. We consider a cut $(V - S, S)$ of $\alpha$, such that we have $s \in V - S$ and $d \in S$. Nodes~$s$ and~$d$ are distinct but we may have no hyperlinks between $s$ and $J$, such that $J = \{s\}$, and also between $K$ and $d$, such that $K = \{d\}$.}
\label{fig:saf_proof}
\end{figure}

\begin{figure*}[ht!]
\centering
\subfigure[] {
  \label{fig:testbed-nodes}
  \includegraphics[width=.45\textwidth]{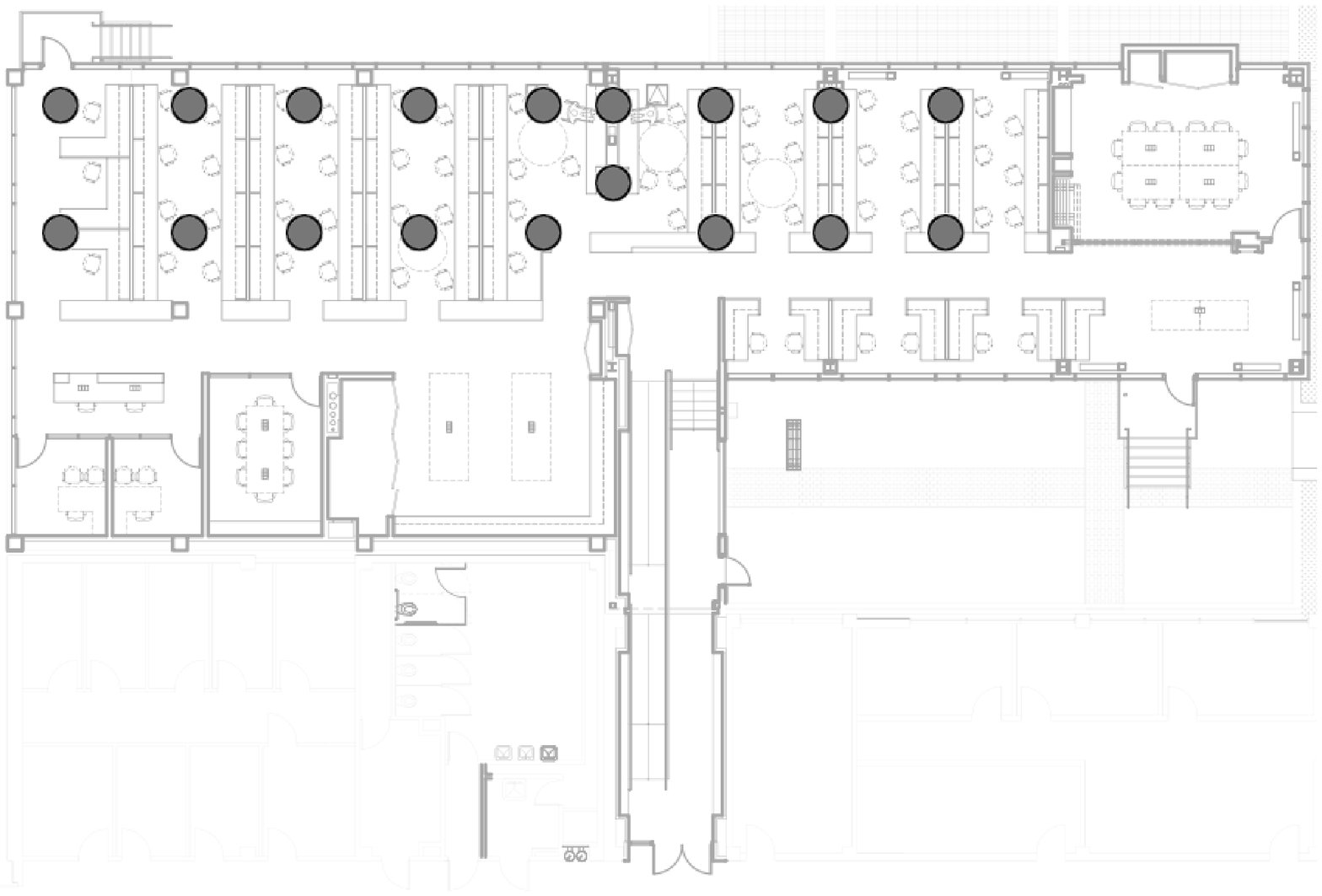}
}
\subfigure[] {
  \label{fig:testbed-delivery}
  \includegraphics[width=.45\textwidth]{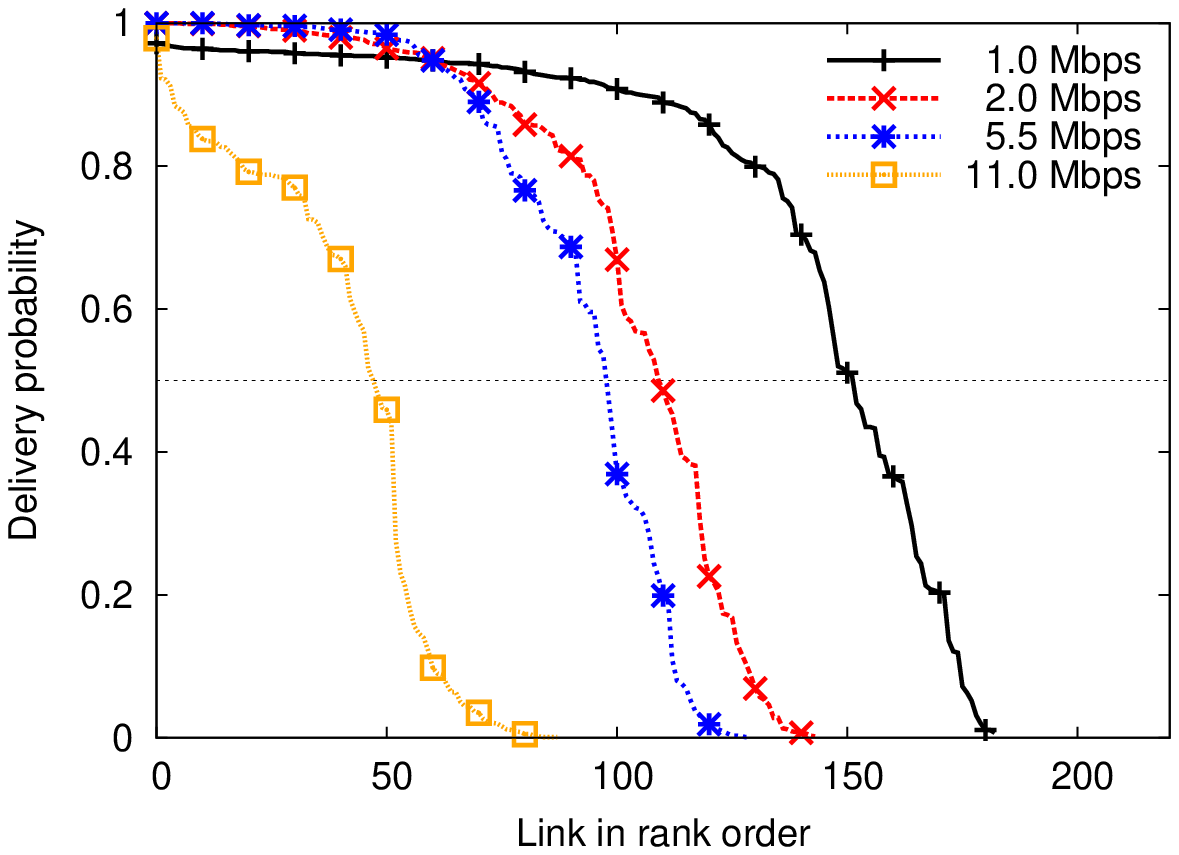}
}
\caption{(a) The location of the nodes in the testbed, arranged in an approximate 2x9 grid. (b) The delivery probabilities of the testbed links for each transmission rate. The data points for each curve are placed in order from largest to smallest (i.e., in rank order). As the rate increases, less links are available and thus path diversity decreases.}
\label{fig:testbed}
\end{figure*}

Without loss of generality, assume that node $i \in J$ has the shortest distance to~$d$ among all nodes in~$V - S$. That is, $\delta_i \leq \delta_j$ for all $j \in V - S$. We claim that every edge leaving node~$i$ must necessarily cross the cut $(V - S, S)$. Thus, for every edge~$(i,j)$ leaving node~$i$, we must have $j \in S$. To prove this claim, let us assume that node~$i$ has an edge $(i,j)$ to another node $j \in V - S$. By Lemma~\ref{lem:acyclic}, we know that in this case we must have $\delta_i \geq \delta_j$. However, since we assumed that node~$i$ has the shortest distance in $V-S$, then $\delta_i \leq \delta_j$ and such an edge could only exist if $\delta_i = \delta_j$. By Lemma~\ref{lem:extra}, we know that if $\delta_i = \delta_j$ then the link $(i,j)$ is redundant and we can safely remove it from the multirate anypath without changing its distance. As a result, for every edge $(i,j)$ we must have $j \in S$. Figure~\ref{fig:saf_proof} shows this situation where node~$i$ only has links to nodes in~$S$.

Additionally, we claim that the nodes in~$S$ were settled in ascending order of distance. That is, if $\delta_j < \delta_k$ then node~$j$ was settled before node~$k$. Since node~$i$ has the shortest distance to~$d$ among all nodes in $V-S$, settling $s$ before~$i$ implies that~$s$ is settled ``out of order.'' For the purpose of contradiction, let~$s$ be the first node settled out of order. This is an assumption which is independent from the initial assumption that $D_s \not = \delta_s$.

We now claim that $D_i = \delta_i$ at the time $s$ is inserted into~$S$. To prove this claim, notice that $K \subseteq S$. Since $s$ is the first node for which $D_s \neq \delta_s$ when it is added to $S$, then we must have $D_k = \delta_k$, for every $k \in K$. Let $\phi_i \subseteq K$ be the forwarding set used in the shortest multirate anypath from~$i$ to~$d$ using the optimal transmission rate $\rho \in R$. 
By Lemma~\ref{lem:order}, $\phi_i$~is composed of the neighbors of~$i$ with the shortest distances to~$d$. Assume that $\phi_i = \{1,2,\ldots,j\}$ with $\delta_1 \leq \delta_2 \leq \ldots \leq \delta_j$. Since $s$ is the first out-of-order node, we know that the nodes in~$S$ were settled in order. Therefore, node~$1$ was settled before node~$2$, which was settled before node~$3$, and so on. At the time node~$1$ is settled, the forwarding set $F_i^{\,(\rho)}$ is initialized to $F_i^{(\rho)}~=~\{1\}$. When node~$2$ is settled, there is no need to check the forwarding set~$\{2\}$. By Lemma~\ref{lem:order}, this forwarding set is never optimal so we just check the set~$\{1,2\}$. By Lemma~\ref{lem:fset}, using $\{1,2\}$ always provides a shorter distance than using just~$\{1\}$. The forwarding set is then updated to~$F_i^{\,(\rho)} = \{1,2\}$. The same procedure is repeated for each settled node, until we finally have $F_i^{\,(\rho)} = \phi_i = \{1,2,\ldots,j\}$. At this time, we also have $D_i^{(\rho)} = \delta_i$, which triggers the update $D_i = D_i^{(\rho)} = \delta_i$, $F_i = F_i^{\,(\rho)} = \phi_i$, and $T_i = \rho$. Once $D_i$ is equal to the shortest distance~$\delta_i$, it does not change anymore and we have $D_i = \delta_i$ at the time $s$ is inserted into $S$.

We can now prove the theorem with two contradictions. Since node~$i$ occurs after node~$s$ in the shortest multirate anypath to~$d$, by Lemma~\ref{lem:acyclic} we have $\delta_i \leq \delta_s$. In addition, we must also have $\delta_s \leq D_s$ because $D_s$ is never smaller than~$\delta_s$. Since both $i$ and $s$ are in $V - S$ and node~$s$ was chosen as the one with the minimum distance from~$Q$, then we must have $D_s \leq D_i$ and $\delta_i \leq \delta_s \leq D_s \leq D_i$. From our previous claim, we know that $D_i = \delta_i$ and therefore $D_i = \delta_i \leq \delta_s \leq D_s \leq D_i$, from which we have
\begin{equation}
D_i = \delta_i = \delta_s = D_s. 
\end{equation}
As a result, $s$ is not settled out of order since~$i$ has the shortest distance in~$V-S$ and $\delta_s = \delta_i$. From this we conclude that the nodes in~$S$ are settled in ascending order of distance. Additionally, we also have $D_s = \delta_s$ at the time $s$ is added to~$S$, which contradicts our initial choice of~$s$. We conclude therefore that for each node~$s \in V$ we have $D_s = \delta_s$ at the time $s$ is added to~$S$.
\end{proof}

\section{Experimental Results}
\label{sec:results}

We evaluated the proposed multirate algorithm using an 18-node 802.11b indoor testbed. Each node is a Stargate microserver~\cite{stargate} equipped with an Intel 400-MHz Xscale PXA255 processor, 64 MB of SDRAM, 32 MB of Flash, and an SMC EliteConnect SMC2532W-B PCMCIA 802.11b wireless network card using the Prism2 chipset. This card has a maximum transmission power of 200~mW and it defaults to a proprietary power control algorithm. The nodes of the testbed are distributed over the ceiling of the Center for Embedded Networked Sensing (CENS) at UCLA. The nodes are located in an approximate 2x9 grid and roughly ten meters apart from each other. Figure~\ref{fig:testbed-nodes} depicts the location of the nodes in the testbed. Each node is equipped with a 3-dB omni-directional rubber duck antenna for the wireless communication. In order to emulate a wireless mesh network with multiple hops, we use a 30-dB SA3-XX attenuator between the wireless interface and its antenna. The attenuator weakens the signal during both the transmission and the reception of a frame, emulating a large distance between nodes. For 11~Mbps, we have paths of up to 8 hops between each pair of nodes, with 3.1 hops on average. For 1~Mbps, we have a longer transmission range, which reduces the maximum path length to 3 hops, with an average of 1.5 hops between each pair of nodes.

We use the testbed to measure the delivery probability of each link at different transmission rates. For that purpose, each node broadcasts one thousand 1500-byte packets and later on we collect the number of received packets at neighbor nodes. We repeat this process for 1, 2, 5.5, and 11~Mbps to have a link estimate for each transmission rate. We use the Click toolkit~\cite{kohler00} and a modified version of the MORE software package~\cite{chachulski07a} for the data collection. Our implementation is capable of sending and receiving raw 802.11 frames by using the wireless network interface in monitor mode. We modified the HostAP Prism driver~\cite{hostap} for Linux in order to allow not only 802.11 frame overhearing but also frame injection while in monitor mode. In addition, we extended the HostAP driver to enable it to control the transmission rate of each 802.11 frame sent. The Click toolkit tags each frame with a selected transmission rate and this information is then passed along to the driver. For each frame, our modification reads the information tagged by Click and notifies the wireless interface firmware about the specified transmission rate. 

Figure~\ref{fig:testbed-delivery} shows the distribution of the delivery probability of each link in the testbed at different 802.11b transmission rates. Every node pair contributes with two links in the graph, one for each direction. Links of each rate are placed in order from largest to smallest (i.e., in rank order). The points of each curve are sorted separately and, therefore, the delivery probabilities of a given x-value are not necessarily from the same link. In wireless mesh networks, higher transmission rates usually have shorter radio ranges and therefore a lower network density. We can see this behavior in Figure~\ref{fig:testbed-delivery}. As the transmission rate increases, we can see that we have less links available and therefore less path diversity between nodes. For instance, as shown by the dashed horizontal line, the number of links with a delivery probability higher than 50\% is 151 at 1~Mbps, 109 at 2~Mbps, 95 at 5.5~Mbps, and only 47 at 11~Mbps. With less paths available at higher rates, we have an interesting tradeoff for multirate anypath routing. With a lower transmission rate, we have more path diversity and a shorter number of hops to traverse, but also a lower throughput. On the other hand, a higher rate results in a higher throughput, but also in less path diversity and a larger number of hops. Our algorithm explores this tradeoff and selects the optimal transmission rate and forwarding set for every node.

Fig.~\ref{fig:independence} shows the results of an experiment we conducted to test the independence of receivers. In our experiment, a node broadcasts 500,000 data frames at 11~Mbps to four neighbors and each frame has 1500 bytes. The x-axis represents the 16 possible set of receivers for the frame (i.e., set 0 corresponds to the frame being lost by all neighbors and set 15 corresponds to every neighbor correctly receiving the frame). The y-axis represents the probability of each set. The ``observed'' histogram is directly derived from the data. The ``independent'' histogram is derived by assuming that the loss probability at each receiver is independent of each other, so it is calculated simply by multiplying the respective probabilities of each individual receiver. We can see that both functions are pretty close indicating that the delivery probabilities of each receiver are loosely correlated. This experiment was repeated for other nodes in the testbed and a similar behavior was observed. Our result are also consistent with other studies~\cite{miu05,reis06}. 

\begin{figure}[ht!]
\centering
\includegraphics[width=.45\textwidth]{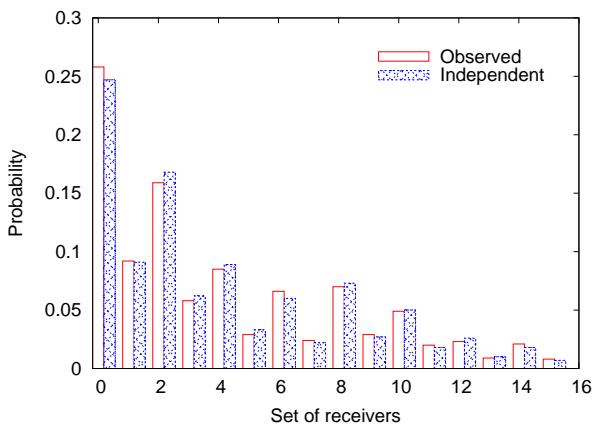}
\caption{(a) Distribution of frame receptions at four neighbors. For four neighbors, we have $2^4 = 16$ subsets and each one represents a different set of neighbors who correctly received the frame. }
\label{fig:independence}
\end{figure}

\begin{figure*}[ht!]
\centering
\subfigure[] {
  \includegraphics[width=.45\textwidth]{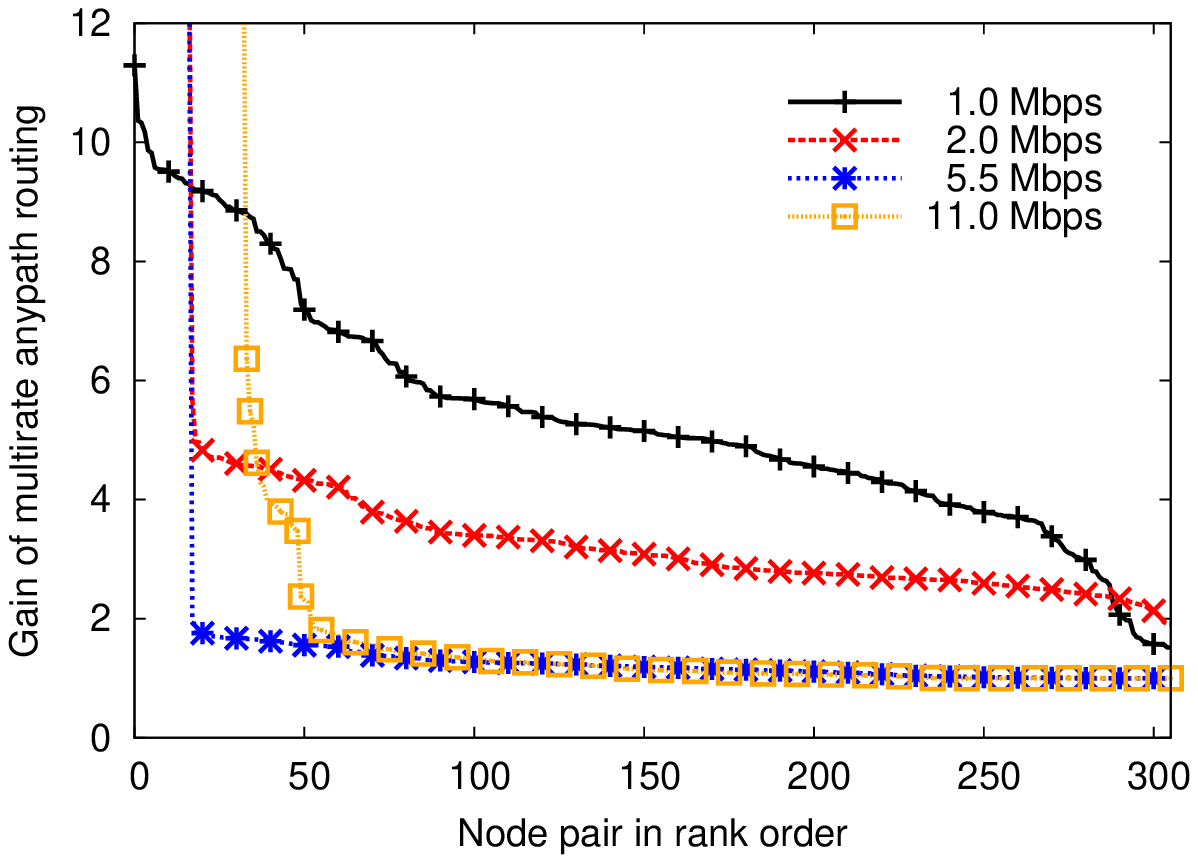}
  \label{fig:gain_ap}
}
\subfigure[] {
  \includegraphics[width=.45\textwidth]{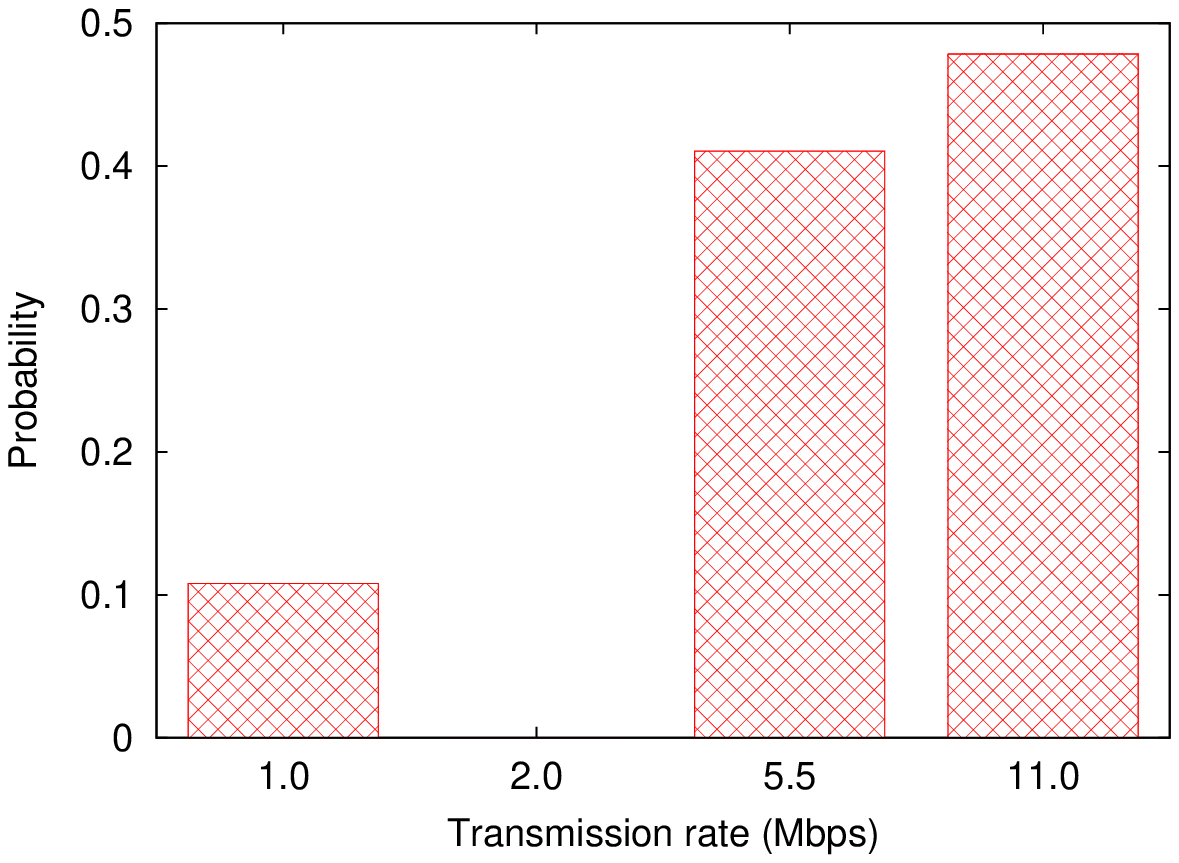}
  \label{fig:rates_pdf}
}
\caption{Results of the SMAF algorithm for the wireless testbed. (a) Gain of multirate over single-rate anypath routing. For each node pair, we indicate in the y-axis how many times multirate anypath routing is better than single-rate anypath. (b) Histogram of the transmission rate chosen by each node. Optimal transmission rates are not concentrated at any particular rate, indicating that a single-rate algorithm can not perform as well as a multirate algorithm.}
\label{fig:results}
\end{figure*}

The shortest multirate anypath {\it always} has an equal or lower cost than the shortest single-rate anypath. Otherwise, we would have a contradiction since we can find another multirate anypath (i.e., the single-rate anypath) with a shorter distance to the destination. It is important, however, to quantify how much better multirate anypath routing is over single-rate anypath. For this purpose, we calculate the gain of multirate over single-rate anypath. We define the gain of a given pair of nodes as the ratio between the single-rate anypath distance and the multirate anypath distance between these two nodes. This metric reflects how many times the end-to-end transmission time is larger when using single-rate as opposed to multirate anypath routing. 

Figure~\ref{fig:gain_ap} shows the distribution of this gain for every pair of nodes in the network. Each curve represents the gain over single-rate anypath routing at a fixed rate. We see that the end-to-end transmission time with multirate anypath routing is at least 50\% and up to 11.3 times shorter than with single-rate anypath routing at 1~Mbps, with an average gain of~5.4. For higher rates, we also see an interesting behavior depicted by the vertical lines. These lines indicate that several node pairs have an {\it infinite gain}. The infinite gain occurs because these nodes can not talk to each other at that particular rate due to the poor link quality; the network therefore becomes disconnected. We have 17 (5.6\%) node pairs that can not reach each other at both 2 and 5.5~Mbps and 33 (10.8\%) node pairs out of reach at 11~Mbps. For the network to be connected, we must then either use a lower rate (e.g., 1~Mbps) for the whole network at the cost of a lower throughput or use multirate anypath routing. For 2~Mbps, if we remove the node pairs with infinite gains, we have a gain of at least 91\% and up to~5.6, with an average of~3.2. For 5.5~Mbps, we have a gain up to 2.0, with an average of 22\%. Finally, for 11~Mbps, we have a gain up to 6.4, with an average of 80\%.

Figure~\ref{fig:rates_pdf} shows the reason why multirate always performs better than single-rate anypath routing. In this graph, we show the distribution of the optimal transmission rates selected by each node to reach every other node. We can see that the optimal transmission rates are not concentrated at a single rate, but rather distributed among over several possibilities. We have 10.8\% of node pairs using 1~Mbps, 41\% using 5.5~Mbps, and 47\% using 11~Mbps as the optimal rate. Interestingly enough, no node pair selected 2~Mbps as the optimal rate since it was more beneficial to use another rate instead. If these rates were concentrated at a particular rate, then multirate and single-rate anypath routing would have the same cost. This assumption, however, does not hold in practice and therefore multirate anypath routing always has a higher performance, sometimes manyfold higher as shown in Figure~\ref{fig:gain_ap}, than single-rate anypath routing.

\section{Related Work}
\label{sec:related}

Most of the work in anypath routing focuses on using a single transmission rate. The following works are all single-rate anypath routing schemes. 

Zorzi and Rao~\cite{zorzi03} use a combination of opportunistic and geographic routing in a wireless sensor network. The authors assume that sensor nodes are aware of their locations and this information is used for routing. The forwarding set of a given node is composed of the neighbors which are physically closer to the destination. Packets are broadcast and neighbors in the set forward the packet respecting the relay priority explained in Section~\ref{sec:anypath}. As an advantage, this routing procedure does not need any sort of route dissemination over the network. Using just the physical distance as the routing metric, however, may not be the best approach since it does not take link quality into account. We introduce the EATT routing metric that takes not only the link quality but also the multiple transmission rates into account during route calculation.

Ye {\it et al.}~\cite{ye05} present another single-rate opportunistic routing protocol for sensor networks. The key idea is that each packet carries a credit which is initially set by the source and is reduced as the packet traverses the network. Each node also maintains a cost for forwarding a packet from itself to the destination and nodes closer to the destination have smaller costs. Packets are sent in broadcast and a neighbor node forwards a received packet only if the credit in the packet is high enough. Just before forwarding the packet, its credit is reduced according to the node cost; therefore, more credits are consumed as the packet moves away from the shortest path. A mesh around the shortest path is then created on-the-fly for each packet. Yuan~{\it et~al.}~\cite{yuan05} use a similar idea for wireless mesh networks. Although packet delivery is improved, this routing scheme increases overhead since it is based on a controlled flooding mechanism. Therefore, robustness comes at the cost of duplicate packets. In our proposal, a packet is forwarded by a single neighbor in the forwarding set and a MAC mechanism, such as the one proposed by Jain and Das~\cite{jain08}, is in place to guarantee that no duplicate packets occur in the network.

Biswas and Morris~\cite{biswas05a} designed and implemented ExOR, an opportunistic routing protocol for wireless mesh networks. ExOR follows the same guidelines of single-rate anypath routing explained in Section~\ref{sec:anypath}. Basically, a node forwards a batch of packets and each neighbor in the forwarding set waits its turn to transmit the received packets. The authors implement a MAC scheduling scheme to enforce the relay priority in the forwarding set. As a result, a node only forwards a packet if all higher priority nodes failed to do so. The authors show that opportunistic routing increases throughput by a factor of two to four compared to single-path routing. Our results go beyond and show that an even better performance can be achieved with multirate anypath routing. Additionally, in our design, each packet is routed independently without storing any per-batch state at intermediate routers.

Chachulski {\it et al.}~\cite{chachulski07a} introduce MORE, a routing protocol which uses both opportunistic routing and network coding to increase the network end-to-end throughput. Upon the receipt of a new packet, a node encodes it with previously received packets and then broadcasts the coded packet. 
Results show that MORE allows a higher throughput than ExOR and single-path routing. Network coding, however, requires routers to store previous packets in order to code them with future packets, adding significant storage and processing overhead to the forwarding process. Furthermore, the authors only focus on opportunistic routing with a single transmission rate. 
Our results indicate that performance could be further improved with multirate anypath routing. An analysis of multirate anypath routing and network coding is also an open problem and an interesting topic for future work. 

Besides using a single bit rate, the above-mentioned systems also do not have a systematic approach for the selecting the forwarding set for a given destination. The selection is commonly based on the heuristic that if a neighbor has a smaller ETX distance to the destination, then it should be in the forwarding set. However, the ETX is a single-path metric and do not represent correctly the node's true distance when using anypath routing. Zhong~{\it et~al.}~\cite{zhong06a} was the first to propose the expected anypath number of transmissions (EATX) metric described in Section~\ref{sec:anypath}, which was also used in~\cite{chachulski07b,dubois-ferriere07}. The authors propose an algorithm for forwarding set selection in~\cite{zhong07}, but this algorithm may not reach an optimal solution depending on the order that neighbors are tested.

Dubois-Ferrière {\it et al.}~\cite{dubois-ferriere07} introduced a shortest anypath algorithm capable of finding optimal forwarding sets. The authors generalize the well-known Bellman-Ford algorithm for anypath routing and prove its optimality. Performance tests in a wireless sensor network show that anypath routing significantly reduces the required number of transmissions from a node to the sink. Chachulski~\cite{chachulski07b} presents a generalization of Dijkstra's algorithm for anypath routing that is very similar to the one we independently derived in Section~\ref{sec:single-rate}, but the author does not provide any proof of optimality. Both of these algorithms, however, are designed for networks using a single transmission rate. Instead, our algorithm in Section~\ref{sec:multi-rate} generalizes anypath routing for multiple rates, giving nodes the ability to choose both the best rate and the best forwarding set to a particular destination. We also provide the proof of optimality for our algorithm. As a result, the optimality of the single-rate algorithm in~\cite{chachulski07b} is also proved since this is a particular case of our algorithm.

More recently, multiple transmission rates have been addressed in opportunistic routing. Radunovic {\it et al.}~\cite{radunovic08} presents an optimization framework to derive routing, scheduling, and rate adaptation schemes. Zeng~{\it et~al.}~\cite{zeng08} presents a linear-programming formulation to optimize the end-to-end throughput of opportunistic routing, considering multiple rates and transmission conflict graphs. However, in both cases the problem being solved is NP-hard. Heuristics are then applied to find a solution, which is not necessarily optimal. 

\section{Conclusions}
\label{sec:conclusions}

In this paper we introduced multirate anypath routing, a new routing paradigm for wireless mesh networks. We provided a solution to integrating opportunistic routing and multiple transmission rates. The available rate diversity imposes several new challenges to routing, since radio range and delivery probabilities change with the transmission rate. Given a network topology and a destination, we want to find both a forwarding set and a transmission rate for every node, such that their distance to the destination is minimized. We pose this as the {\it shortest multirate anypath problem}. Finding the rate and forwarding set that jointly optimize the distance from a node to a given destination is considered an open problem. To solve it, we introduced the EATT routing metric as well as the Shortest Multirate Anypath First (SMAF) algorithm and presented a proof of its optimality. Our algorithm has the same complexity as Dijkstra's algorithm for multirate single-path routing, being easy to implement in link-state routing protocols.

We conducted experiments in a 18-node 802.11b testbed to evaluate the performance of multirate over single-rate anypath routing. Our main findings are: (1) when the network uses a single bit rate, it may become disconnected since some links may not work at the selected rate; (2) multirate outperforms 11-Mbps anypath routing by 80\% on average and up to a factor of 6.4 while still maintaining full connectivity; (3) multirate also outperforms 1-Mbps anypath routing by a factor of 5.4 on average and up to a factor of 11.3; (4) the distribution of the optimal transmission rates are not concentrated at any particular rate, corroborating the assumption that hyperlinks in single-rate anypath routing usually do not transmit at their optimal rates. 

\section*{Acknowledgments}
This work was done in part while the first author was visiting the Ecole Polytechnique Fédérale de Lausanne (EPFL). We would like to thank Henri Dubois-Ferrière and Martin Vetterli for hosting the first author at EPFL and introducing anypath routing to him. We thank Deborah Estrin for her help and discussions over the years and for the CENS testbed. We also thank Eddie Kohler, Fan Ye, and Lixia Zhang for insightful comments on an early draft. We are grateful to Martin Lukac for his help with the testbed. This work was supported by the U.S. National Science Foundation under Grants NBD-0721963 and CCF-0120778. Any opinions, findings, and conclusions or recommendations expressed in this material are those of the authors and do not necessarily reflect the views of the National Science Foundation.

\bibliographystyle{IEEEtran}
\bibliography{arxiv}

\begin{thebibliography}{10}
\providecommand{\url}[1]{#1}
\csname url@samestyle\endcsname
\providecommand{\newblock}{\relax}
\providecommand{\bibinfo}[2]{#2}
\providecommand{\BIBentrySTDinterwordspacing}{\spaceskip=0pt\relax}
\providecommand{\BIBentryALTinterwordstretchfactor}{4}
\providecommand{\BIBentryALTinterwordspacing}{\spaceskip=\fontdimen2\font plus
\BIBentryALTinterwordstretchfactor\fontdimen3\font minus
  \fontdimen4\font\relax}
\providecommand{\BIBforeignlanguage}[2]{{%
\expandafter\ifx\csname l@#1\endcsname\relax
\typeout{** WARNING: IEEEtran.bst: No hyphenation pattern has been}%
\typeout{** loaded for the language `#1'. Using the pattern for}%
\typeout{** the default language instead.}%
\else
\language=\csname l@#1\endcsname
\fi
#2}}
\providecommand{\BIBdecl}{\relax}
\BIBdecl

\bibitem{aguayo04}
D.~Aguayo, J.~Bicket, S.~Biswas, G.~Judd, and R.~Morris, ``{Link-level
  Measurements from an 802.11b Mesh Network},'' in \emph{{Proceedings of the
  ACM SIGCOMM'04 Conference}}, {Portland, OR, USA}, Aug. 2004, pp. 121--131.

\bibitem{campista08}
M.~Campista, P.~Esposito, I.~Moraes, L.~H. Costa, O.~C. Duarte, D.~Passos,
  C.~V. de~Albuquerque, D.~C. Saade, and M.~Rubinstein, ``{Routing Metrics and
  Protocols for Wireless Mesh Networks},'' \emph{{IEEE Network}}, vol.~22,
  no.~1, pp. 6--12, Jan.-Feb. 2008.

\bibitem{seth07}
A.~Seth, D.~Kroeker, M.~Zaharia, S.~Guo, and S.~Keshav, ``{Low-cost
  Communication for Rural Internet Kiosks using Mechanical Backhaul},'' in
  \emph{Proceedings of the ACM Mobicom'06 Conference}, Los Angeles, CA, USA,
  Sep. 2006, pp. 334--345.

\bibitem{chachulski07b}
S.~Chachulski, ``{Trading Structure for Randomness in Wireless Opportunistic
  Routing},'' Master's thesis, Massachusetts Institute of Technology,
  Cambridge, MA, USA, May 2007.

\bibitem{biswas05a}
S.~Biswas and R.~Morris, ``{ExOR: Opportunistic Multi-Hop Routing for Wireless
  Networks},'' in \emph{{Proceedings of the ACM SIGCOMM'05 Conference}},
  {Philadelphia, PA, USA}, Aug. 2005, pp. 133--143.

\bibitem{zhong06a}
Z.~Zhong, J.~Wang, S.~Nelakuditi, and G.-H. Lu, ``{On Selection of Candidates
  for Opportunistic AnyPath Forwarding},'' \emph{{ACM SIGMOBILE Mobile
  Computing and Communications Review}}, vol.~10, no.~4, pp. 1--2, Oct. 2006.

\bibitem{dubois-ferriere07}
H.~Dubois-Ferriere, M.~Grossglauser, and M.~Vetterli, ``{Least-Cost
  Opportunistic Routing},'' in \emph{{Proceedings of the 2007 Allerton
  Conference}}, {Monticello, IL, USA}, Sep. 2007.

\bibitem{zeng08}
K.~Zeng, W.~Lou, and H.~Zhai, ``{On End-to-End Throughput of Opportunistic
  Routing in Multirate and Multihop Wireless Networks},'' in \emph{Proceedings
  of the IEEE Infocom'08}, Phoenix, AZ, USA, Apr. 2008, pp. 816--824.

\bibitem{draves04b}
R.~Draves, J.~Padhye, and B.~Zill, ``{Routing in Multi-Radio, Multi-Hop
  Wireless Mesh Networks},'' in \emph{{Proceedings of the ACM MobiCom'04
  Conference}}, {Philadelphia, PA, USA}, Sep. 2004, pp. 114--128.

\bibitem{jain08}
S.~Jain and S.~R. Das, ``{Exploiting Path Diversity in the Link Layer in
  Wireless Ad Hoc Networks},'' \emph{{Ad Hoc Networks}}, vol.~6, no.~5, pp.
  805--825, Jul. 2008.

\bibitem{reis06}
C.~Reis, R.~Mahajan, M.~Rodrig, D.~Wetherall, and J.~Zahorjan,
  ``{Measurement-Based Models of Delivery and Interference in Static Wireless
  Networks},'' in \emph{{Proceedings of the ACM SIGCOMM'06 Conference}}, {Pisa,
  Italy}, Sep. 2006, pp. 51--62.

\bibitem{miu05}
A.~Miu, H.~Balakrishnan, and C.~E. Koksal, ``{Improving Loss Resilience with
  Multi-Radio Diversity in Wireless Networks},'' in \emph{Proceedings of the
  ACM Mobicom'05}, Cologne, Germany, Aug. 2005, pp. 16--30.

\bibitem{zorzi03}
M.~Zorzi and R.~R. Rao, ``{Geographic Random Forwarding (GeRaF) for Ad Hoc and
  Sensor Networks: Multihop Performance},'' \emph{{IEEE Transactions on Mobile
  Computing}}, vol.~2, no.~4, pp. 337--348, Oct.-Dec. 2003.

\bibitem{couto03}
D.~D. Couto, D.~Aguayo, J.~Bicket, and R.~Morris, ``{A High-Throughput Path
  Metric for Multi-Hop Wireless Routing},'' in \emph{{Proceedings of the ACM
  MobiCom'03 Conference}}, {San Diego, CA, USA}, Sep. 2003, pp. 134--146.

\bibitem{chachulski07a}
S.~Chachulski, M.~Jennings, S.~Katti, and D.~Katabi, ``{Trading Structure for
  Randomness in Wireless Opportunistic Routing},'' in \emph{{Proceedings of the
  ACM SIGCOMM'07 Conference}}, {Kyoto, Japan}, Aug. 2007.

\bibitem{cormen01}
T.~T. Cormen, C.~E. Leiserson, and R.~L. Rivest, \emph{Introduction to
  algorithms}, 2nd~ed.\hskip 1em plus 0.5em minus 0.4em\relax Cambridge, MA,
  USA: MIT Press, 2001.

\bibitem{stargate}
``{PlatformX with Stargate},'' http://platformx.sourceforge.net.

\bibitem{kohler00}
E.~Kohler, R.~Morris, B.~Chen, J.~Jannotti, and M.~F. Kaashoek, ``The click
  modular router,'' \emph{ACM Transactions on Computer Systems}, vol.~18,
  no.~3, pp. 263--297, 2000.

\bibitem{hostap}
``{Host AP driver for Intersil Prism2/2.5/3, hostapd, and WPA Supplicant},''
  http://hostap.epitest.fi.

\bibitem{ye05}
F.~Ye, G.~Zhong, S.~Lu, and L.~Zhang, ``{GRAdient Broadcast: A Robust Data
  Delivery Protocol for Large Scale Sensor Networks},'' \emph{Wireless
  Networks}, vol.~11, no.~3, pp. 285--298, 2005.

\bibitem{yuan05}
Y.~Yuan, H.~Yang, S.~H.~Y. Wong, S.~Lu, and W.~Arbaug, ``{ROMER: Resilient
  Opportunistic Mesh Routing for Wireless Mesh Networks},'' in
  \emph{{Proceedings of the IEEE Workshop on Wireless Mesh Networks
  (WiMesh'05)}}, {Santa Clara, CA, USA}, Sep. 2005.

\bibitem{zhong07}
Z.~Zhong and S.~Nelakuditi, ``{On the Efficacy of Opportunistic Routing},'' in
  \emph{{Proceedings of the IEEE SECON'07 Conference}}, {San Diego, CA, USA},
  Jun. 2007.

\bibitem{radunovic08}
B.~Radunovic, C.~Gkantsidis, P.~Key, and P.~Rodriguez, ``{An Optimization
  Framework for Opportunistic Multipath Routing in Wireless Mesh Networks},''
  in \emph{Proceedings of the IEEE Infocom'08}, Phoenix, AZ, USA, Apr. 2008,
  pp. 2252--2260.

\end{thebibliography}

\appendices
\section{Proofs of the Lemmas}
\label{app:proofs}

{\it Lemma~\ref{lem:comparison}: For a fixed transmission rate, let $D_i$ be the distance of a node~$i$ via forwarding set~$J$ and let $D_i'$ be the distance via forwarding set $J' = J \cup \{n\}$, where $D_n \geq D_j$ for every node $j \in J$. We have $D_i' \leq D_i$ if and only if $D_i \geq D_n$.}
\begin{proof}
Assume that $D_1 \leq D_2 \leq \ldots \leq D_{n-1} \leq D_n$ and $J=\{1,2,\ldots,n-1\}$. Let $D_i = d_{iJ} + D_J$ be the distance from node~$i$ using the forwarding set~$J$. From~(\ref{eq:D_J_def}) and (\ref{eq:w_j_def}), the remaining-anypath cost $D_J$ is defined as
\begin{equation}
\label{eq:D_J_def2} D_J = \frac{\displaystyle\sum_{j=1}^{n-1} p_jD_j\prod_{k=1}^{j-1}(1-p_k)}{\displaystyle1-\prod_{j \in J}(1-p_j)}.
\end{equation}
Let $D_i' = d_{iJ'} + D_{J'}$ be this distance via $J'=J \cup \{n\}$, where
\begin{equation}
\label{eq:D_J'_def} D_{J'} = \frac{\displaystyle\sum_{j=1}^{n} p_jD_j\prod_{k=1}^{j-1}(1-p_k)}{\displaystyle1-\prod_{j \in J'}(1-p_j)}.
\end{equation}
If we define the probabilities $p_J$ and $p_{J'}$ as
\setlength{\arraycolsep}{0.0em}%
\begin{eqnarray}
\nonumber p_J &{}={}& 1 - \prod_{j \in J}  (1 - p_j) \\
p_{J'} &{}={}& 1 - \prod_{j \in J'} (1 - p_j) = 1 - (1-p_J)(1-p_n),
\end{eqnarray}
\setlength{\arraycolsep}{5pt}%
we can rewrite $D_{J'}$ in~(\ref{eq:D_J'_def}) in terms of $D_J$ in~(\ref{eq:D_J_def2}) as
\begin{equation}
\label{eq:aggregated} D_{J'} = \frac{p_JD_J + \left(1-p_J\right)p_nD_n}{1-\left(1-p_J\right)\left(1-p_n\right)}.
\end{equation}
An interesting result from~(\ref{eq:aggregated}) is that we can see the forwarding set $J$ as an ``aggregated vertex'' with delivery probability~$p_J$ and distance~$D_J$.

We now show that if $D_i \geq D_n$, then $D_i' \leq D_i$ as follows
\setlength{\arraycolsep}{0.0em}%
\begin{eqnarray}
\nonumber D_i &{}\geq{}& D_n \\
\nonumber \left(1-\frac{p_J}{p_{J'}}\right)D_i &{}\geq{}& \left(1-\frac{p_J}{p_{J'}}\right)D_n \\
\nonumber \left(1-\frac{p_J}{p_{J'}}\right)\left(d_{iJ} + D_J\right) &{}\geq{}& \frac{\left(1-p_J\right)p_n}{p_{J'}}D_n \\
\nonumber d_{iJ} + D_J &{}\geq{}& \frac{p_J}{p_{J'}}\left(d_{iJ} + D_J\right) + \frac{\left(1-p_J\right)p_n}{p_{J'}}D_n \\
\nonumber d_{iJ} + D_J &{}\geq{}& \frac{p_J}{p_{J'}}d_{iJ} + \frac{p_JD_J + \left(1-p_J\right)p_nD_n}{1-(1-p_J)(1-p_n)} \\
\nonumber d_{iJ} + D_J &{}\geq{}& d_{iJ'} + D_{J'} \\
\label{eq:D_i_and_D_i'} D_i &{}\geq{}& D_i'.
\end{eqnarray}
\setlength{\arraycolsep}{5pt}%
To show that if $D_i' \leq D_i$ then $D_i \geq D_n$, we just take~(\ref{eq:D_i_and_D_i'}) in the reverse order. Consequently, if $D_i > D_n$, it is better to use the forwarding set $J'= J \cup \{n\}$ instead of~$J$, since the distance~$D_i'$ via~$J'$ is always shorter than~$D_i$ via~$J$.
\end{proof}

\medskip
{\it Lemma~\ref{lem:acyclic}: The shortest distance~$\delta_i$ of a node~$i$ is always larger than or equal to the shortest distance $\delta_j$ of any node~$j$ in the optimal forwarding set $\phi_i$. That is, we have $\delta_i \geq \delta_j$ for all $j \in \phi_i$.}
\begin{proof}
Let $\phi_i = \{1,2,\ldots,n\}$ with $\delta_1 \leq \delta_2 \leq \ldots \leq \delta_n$ and let $D_i \geq \delta_i$ be the distance via the suboptimal forwarding set $J = \{1,2,\ldots,n-1\}$ with the same transmission rate. From Lemma~\ref{lem:comparison}, we know that if $\delta_i \leq D_i$, then $D_i \geq \delta_n$. From this, we show that $\delta_i \geq \delta_n$ as follows (assume $J'= \phi_i$)
\setlength{\arraycolsep}{0.0em}%
\begin{eqnarray}
\nonumber D_i &{}\geq{}& \delta_n \\
\nonumber \left(\frac{p_J}{p_{J'}}\right)D_i &{}\geq{}& \left(\frac{p_J}{p_{J'}}\right)\delta_n \\
\nonumber \left(\frac{p_J}{p_{J'}}\right)\left(d_{iJ} + D_J\right) &{}\geq{}& \left(\frac{p_J}{p_{J'}}\right)\delta_n \\
\nonumber \frac{p_J}{p_{J'}}d_{iJ} + \frac{p_J}{p_{J'}}D_J &{}\geq{}& \delta_n - \frac{\left(1-p_J\right)p_n}{p_{J'}}\delta_n \\
\nonumber \frac{p_J}{p_{J'}}d_{iJ} + \frac{p_JD_J + \left(1-p_J\right)p_n\delta_n}{p_{J'}} &{}\geq{}& \delta_n \\
\nonumber d_{iJ'} + D_{J'} &{}\geq{}& \delta_n \\
\delta_i &{}\geq{}& \delta_n.
\end{eqnarray}
\setlength{\arraycolsep}{5pt}%
Since $\delta_n$ is the largest distance in the optimal forwarding set~$\phi_i$,  then we know that if $\delta_i \geq \delta_n$ we must have $\delta_i \geq \delta_j$ for all $j \in \phi_i$.
\end{proof}

\medskip
{\it Lemma~\ref{lem:extra}: For any transmission rate, if a node~$i$ uses a node~$n$ in its optimal forwarding set~$\phi_i$ and $\delta_i = \delta_n$, we can safely remove~$n$ from~$\phi_i$ without changing~$\delta_i$. The link $(i,n)$ is said to be ``redundant.''}
\begin{proof}
By Lemma~\ref{lem:acyclic}, we know that $\delta_i \geq \delta_j$, for all $j \in \phi_i$. Since $\delta_i = \delta_n$, we also know that $\delta_n$ is the largest distance in the forwarding set. Let $\phi_i = \{1,2,\ldots,n\}$ with distances $\delta_1 \leq \delta_2 \leq \ldots \leq \delta_n$ and let $D_i = d_{iJ} + D_J$ be the distance from node~$i$ via forwarding set~$J = \{1,2,\ldots,n-1\}$. We now show that if $\delta_i = \delta_j$, then $D_i = \delta_i$ as follows
\setlength{\arraycolsep}{0.0em}%
\begin{eqnarray}
\nonumber \delta_i &{}={}& \delta_n \\
\nonumber d_{iJ'} + D_{J'} &{}={}& \delta_n \\
\nonumber d_{iJ'} + \frac{p_JD_J + \left(1-p_J\right)p_n\delta_n}{1-(1-p_J)(1-p_n)} &{}={}& \delta_n \\
\nonumber d_{iJ'} + \frac{p_J}{p_{J'}}D_J &{}={}& \delta_n - \frac{\left(1-p_J\right)p_n}{p_{J'}}\delta_n\\
\nonumber d_{iJ'} + \frac{p_J}{p_{J'}}D_J &{}={}& \frac{p_J}{p_{J'}}\delta_n\\
\nonumber \frac{p_{J'}}{p_J}d_{iJ'} + D_J &{}={}& \delta_n\\
\nonumber d_{iJ} + D_J &{}={}& \delta_n\\
D_i &{}={}& \delta_n.
\end{eqnarray}
\setlength{\arraycolsep}{5pt}%
Since $D_i = \delta_n$, the forwarding set~$J$ is also optimal and yields the same distance as $\phi_i$. We say the link $(i,n)$ is ``redundant'' since it does not help to reduce the distance any further.
\end{proof}

\medskip
{\it Lemma~\ref{lem:order}: If the shortest distances from the neighbors of a node~$i$ to a given destination are $\delta_1 \leq \delta_2 \leq \ldots \leq \delta_n$, then~$\phi_i^{(r)}$ is always of the form $\phi_i^{(r)} = \{1,2,\ldots,k\}$, for some $k \in \{1,2,\ldots,n\}$.}
\begin{proof}
We prove this lemma by reverse induction. In the basis step, we show that the forwarding set $\{k-1,k\}$ provides a shorter distance than $\{k\}$. In the inductive step, we prove that the forwarding set $\{j-1,j,j+1,\ldots,k-1,k\}$ always provides a shorter distance than $\{j,j+1,\ldots,k-1,k\}$.

{\bf Basis.} Let~$k$ be the node in~$\phi_i^{(r)}$ with the longest distance to the destination and let $D_i = d_{iJ} + D_{J}$ be the distance of node~$i$ to the destination via forwarding set $J=\{k\}$ with 
\begin{equation}
\label{eq:DJ} D_J = D_k. 
\end{equation}
In addition, let $D'_i = d_{iJ'} + D_{J'}$ be the distance of node~$i$ via forwarding set $J' =\{k-1,k\}$, where
\begin{equation}
\label{eq:DJ'} D_{J'} = \frac{p_{k-1}D_{k-1} + (1-p_{k-1})p_kD_k}{1-(1-p_{k-1})(1-p_k)}.
\end{equation}
We can see in $(\ref{eq:DJ'})$ that $D_{J'}$ is a weighted average between $D_{k-1}$ and $D_k$, with the weights summing to one. Comparing with~$(\ref{eq:DJ})$, we can see that in~$(\ref{eq:DJ'})$ we are moving some weight from~$D_k$ to~$D_{k-1}$. Since $D_{k-1} \leq D_k$, we have $D_{J'} \leq D_J$. For the EATX and EATT metric, we also have that $d_{iJ'} \leq d_{iJ}$ for $J \subseteq J'$. As a result, we have
\setlength{\arraycolsep}{0.0em}%
\begin{eqnarray}
\nonumber d_{iJ'} + D_{J'} &{}\leq{}&  d_{iJ} + D_J \\
\label{eq:D_i'_and_D_i} D_i' &{}\leq{}& D_i.
\end{eqnarray}
\setlength{\arraycolsep}{5pt}%
Therefore, the forwarding set $J'= \{k-1,k\}$ always provides \newpage\noindent a shorter distance than the forwarding set $J = \{k\}$. 

{\bf Inductive step.} \hspace{1mm} Let the forwarding sets $J$ and $J'$ be now redefined for the inductive step as $J = \{j,j+1,\ldots,k-1,k\}$ and $J' = \{j-1\} \cup J$. Let $D_i = d_{iJ} + D_J$ be the distance of node~$i$ via forwarding set~$J$ and let $D'_i = d_{iJ'} + D_{J'}$ be the distance of node~$i$ via forwarding set~$J'$. From~(\ref{eq:aggregated}), we can consider the forwarding set $J$ as an ``aggregated node'' with delivery probability $p_J$ and distance $D_J$. We then write $D_{J'}$ in terms of $D_J$ as
\begin{equation}
\label{eq:DJ'3} D_{J'} = \frac{p_{j-1}D_{j-1} + \left(1-p_{j-1}\right)p_JD_J}{1-\left(1-p_{j-1}\right)\left(1-p_J\right)}.
\end{equation}

By our definition of the remaining cost in~(\ref{eq:D_J_def}), $D_J$ is a weighted average of $D_j,D_{j+1},\ldots,D_{k-1},D_k$, which are all larger than $D_{j-1}$. We can see that in~$(\ref{eq:DJ'3})$ we are moving some weight from~$D_J$ to~$D_{j-1}$. Since $D_{j-1} \leq D_J$, we have $D_{J'} \leq D_J$. For the EATX and EATT metric, we also have that $d_{iJ'} \leq d_{iJ}$ for $J \subseteq J'$. Consequently, we have the same result of~(\ref{eq:D_i'_and_D_i}).

We know that the set $J'=\{j-1,j,j+1,\ldots,k-1,k\}$ is always better than $J=\{j,j+1,\ldots,k-1,k\}$. Combining the results from the basis and the inductive steps, we conclude that $\phi_i^{(r)}$ is always of the form $\phi_i^{(r)} = \{1,2,\ldots,k\}$, for some $k \in \{1,2,\ldots,n\}$.
\end{proof}

\medskip
{\it Lemma~\ref{lem:fset}: For a given transmission rate $r \in R$, assume that $\phi_i^{(r)} = \{1,2,\ldots,n\}$ with distances $\delta_1 \leq \delta_2 \leq \ldots \leq \delta_n$. If $D_i^j$ is the distance from node~$i$ using transmission rate~$r$ via forwarding set $\{1,2,\ldots,j\}$, for $1 \leq j \leq n$, then we always have $D_i^1 \geq D_i^2 \geq \ldots \geq D_i^n = \delta_i^{(r)}$.}
\begin{proof}
We want to prove that the forwarding set $J = \{1\}$ yields a larger distance than $J' = \{1,2\}$, which yields a larger distance than $J'' = \{1,2,3\}$ and so on until we get to the forwarding set $\phi_i^{(r)} = \{1,2,\ldots,n\}$. Thus far, we have
\begin{equation}
\label{eq:so_far} \delta_{n-1} \stackrel{(a)}{\leq} \delta_n \stackrel{(b)}{\leq} \delta_i \stackrel{\raisebox{.125cm}{\scriptsize $(c)$}}{=} D_i^n \stackrel{(d)}{\leq} D_i^{n-1}, 
\end{equation}
where $(a)$ and $(c)$ hold by definition, $(b)$ holds by Lemma~\ref{lem:acyclic}, and $(d)$ holds because $\phi_i^{(r)} = \{1,2,\ldots,n\}$ yields the shortest distance to the destination at rate~$r$. 

We now extend this result further for other forwarding sets. We first claim that
\begin{equation}
\delta_{n-2} \stackrel{(a)}{\leq} \delta_{n-1} \leq \delta_n \leq \delta_i = D_i^n \leq D_i^{n-1} \stackrel{(b)}{\leq} D_i^{n-2},
\end{equation}
where $(a)$ holds by definition and $(b)$ holds because of the following argument. By definition, we have $\delta_i = D_i^n \leq D_i^{n-2}$, since $D_i^n$ is the shortest distance to the destination. From~(\ref{eq:so_far}), we then have that $\delta_{n-1} \leq D_i^{n-2}$. Finally, if $\delta_{n-1} \leq D_i^{n-2}$, then $D_i^{n-1} \leq D_i^{n-2}$ by Lemma~\ref{lem:comparison}. 

The same argument can be made recursively until we get
\begin{equation}
\hspace{-0mm} \delta_{1} \leq \ldots \leq \delta_{n-1} \leq \delta_n \leq \delta_i = D_i^n \leq D_i^{n-1} \leq \ldots \leq D_i^{1}, \hspace{-4mm}
\end{equation}
from which we know $D_i^1 \geq D_i^2 \geq \ldots \geq D_i^n$ must be true.
\end{proof}

\end{document}